\documentclass[aps,prl,twocolumn,superscriptaddress,10pt,article,showpacs,longbibliography]{revtex4-2}

\usepackage{blindtext}
\usepackage{lipsum}
\usepackage{graphics}
\usepackage{amsmath}
\usepackage{graphicx}
\usepackage{graphics}
\usepackage{amssymb}
\usepackage{pifont}
\usepackage{verbatim}
\usepackage{physics}
\usepackage{float}
\usepackage[normalem]{ulem}
\usepackage[dvipsnames]{xcolor}
\usepackage{bbm}
\usepackage{enumitem} 
\usepackage{multirow}

\usepackage{graphicx}
\usepackage{bm}
\usepackage{amsmath}
\usepackage{physics}
\usepackage{amssymb}
\usepackage{amsfonts}
\usepackage{amsthm}
\usepackage{bbm}
\usepackage{mathtools}
\usepackage{braket}
\usepackage[normalem]{ulem}
\usepackage{wrapfig}
\usepackage{tikz}
\usepackage{dsfont}
\usepackage{comment}
\usepackage{thmtools,thm-restate}
\usepackage[export]{adjustbox}

\definecolor{blueviolet}{rgb}{0.2, 0.2, 0.6}
\definecolor{webgreen}{rgb}{0,.5,0}
\definecolor{webbrown}{rgb}{.6,0,0}
\usepackage[bookmarks=false,
colorlinks=true, 
urlcolor=webbrown, 
linkcolor=blueviolet, 
citecolor=webgreen,
pdfstartpage=1,
pdfstartview={FitH},  
bookmarksopen=false
]{hyperref}

\newtheorem{theorem}{Theorem}
\newtheorem{definition}{Definition}
\newtheorem{corollary}{Corollary}
\newtheorem{lemma}{Lemma}

\newtheorem*{theorem*}{Theorem}

\newtheorem*{task*}{Task}
\newtheorem*{proposition*}{Proposition}

\newcommand{\dif}{{\rm d}}

\newcommand{\bs}{\boldsymbol}
\newcommand{\be}{\begin{equation}}
	\newcommand{\ee}{\end{equation}}

\DeclareMathOperator*{\E}{{\mathbb{E}}}

\newcommand{\iu}{{\rm i}}

\begin{document}
	
	\title{The Scrooge ensemble in many-body quantum systems}
	
	\author{Max McGinley}
	\affiliation{T.C.M.~Group, Cavendish Laboratory, University of Cambridge, Cambridge CB3 0HE, UK}
	
	\author{Thomas Schuster}
	\affiliation{Walter Burke Institute for Theoretical Physics and Institute for Quantum Information and Matter, California Institute of Technology, Pasadena, California 91125, USA}
	\affiliation{Google Quantum AI, Venice, California 90291, USA}

	\begin{abstract}
		In many physical settings, the statistical properties of quantum states are thought to be described by the Scrooge ensemble, a more structured generalization of the Haar ensemble. In this work, we prove several key results on the properties and complexity of Scrooge-random states in macroscopic quantum systems, and provide a general-purpose calculus for evaluating their moments. A key theme of our results is a separation between universal random fluctuations in non-local properties and exponential concentration of all local properties. Implications for device benchmarking, sampling advantages beyond random circuits, quantum complexity growth, and the physical origin of Scrooge-random states are discussed.
	\end{abstract} 
	
	\maketitle

	Despite the extraordinary complexity that arises in many-body quantum systems, their universal properties can often be described by simple statistical theories---an idea that dates back to Wigner \cite{wigner1955characteristic}.
	By identifying an appropriate ensemble of systems wherein averaged properties can be computed, quantitative predictions about a particular system of interest can be made without needing to analyze any single, often intractable, instance.
	
	In the context of dynamics, the spherical Haar ensemble---a uniform distribution over all states in Hilbert space---is expected to describe the late-time states of scrambling systems without any conservation laws (e.g.~energy, charge, etc.)~\cite{Page93,hayden2007black,cotler2017chaos,roberts2017chaos}.  By understanding the key properties of Haar-random states, such as the structure of bipartite entanglement \cite{Page93}, the thermalization of subsystems \cite{popescu2006entanglement}, and the statistics of measurement outcome distributions \cite{boixo2018characterizing}, precise predictions can be made about a diverse range of physical systems in a unified way. For example, states generated by local random quantum circuits converge towards Haar-random behavior, in a sense made precise by the study of their low-order moments \cite{emerson2003pseudo,harrow2009random,brown2012scrambling, brandao2016local,haferkamp2022random, chen2024incompressibility,laracuente2024approximate,schuster2024random,laracuente2025quantum}. In addition to providing a tractable toy model for understanding scrambling \cite{nahum2017entgrowth,nahum2018operator,von2018operator}, the emergence of Haar-randomness in these circuits has facilitated applications in device benchmarking \cite{emerson2005scalable,boixo2018characterizing}, quantum learning \cite{huang2020predicting,aharonov2022quantum,huang2021quantum,chen2021exponential}, cryptography \cite{ji2018pseudorandom,ma2024construct,schuster2025strong}, and quantum sampling advantages \cite{arute2019quantum,movassagh2023hardness,bouland2019complexity}.
	
	At the same time, the Haar ensemble is limited in its ability to describe most physical systems due to its complete lack of structure. Many systems cannot scramble towards Haar-random behavior, e.g.~due to constraints imposed by conserved charges or energy, yet nonetheless  exhibit many hallmarks of scrambling and chaos~\cite{srednicki1994chaos,rigol2008thermalization,rakovszky2018diffusive,khemani2018operator,cui2025random}. Given the impact that the ensemble approach has had on our understanding of scrambling in unstructured systems, a similar statistical description of these more structured settings is highly desirable. 

	Building on earlier studies~\cite{goldstein2006distribution,goldstein2016universal,cotler2021emergent}, a potential breakthrough in this direction emerged in a recent work~\cite{mark2024maximum}, which suggested that such a description is provided by the so-called Scrooge ensemble, studied previously for few-body quantum systems several decades ago~\cite{jozsa1994lower}.
	This proposal seems imminently plausible: The central assumption of scrambling is that a state looks maximally random up to some global physical constraints; the Scrooge ensemble provides precisely such a way to specify this additional structure.
	However, unlike for Haar-random states, our knowledge of the universal properties of Scrooge-random quantum states is much less known.
	This stems not only from limited explorations, but also from an apparent intractability of computing ensemble-averaged quantities. This significantly limits the analytic power of the  ensemble approach as it stands.

	In this work, we overcome this challenge by developing a simple and rigorous framework for computing averages over Scrooge ensembles in many-body quantum systems, and use this framework to establish several key results on the properties and complexity of Scrooge-random states. Our main contributions are fourfold:

	First, our key technical contribution is a general-purpose approximate formula for the moments of any Scrooge ensemble, with rigorous accuracy guarantees (Theorem \ref{thm:approx}). Our formula is inspired by a proposition made for the second moment in~\cite{mark2024maximum}. We prove that our approximate formula holds up to an exponentially small relative error: the strongest notion of approximation borrowed from the unitary design literature~\cite{brandao2016local,schuster2024random}.

	Second, building on this formula, we show that the moments of the Scrooge ensemble concentrate tightly around their mean on any subsystem missing only a small $\mathcal{O}(\log n)$ number of sites---a vanishing fraction of the entire system. 
	This suppression of fluctuations parallels recent results on random circuits and Hamiltonian dynamics~\cite{HaydenLeungWinter2006,cotler2022fluctuations,huang2019instability,shaw2025experimental,tang2025estimating}, and is much more prominent than in traditional perspectives on thermalization~\cite{abanin2019colloquium}, which usually apply only to sub-extensive subsystems~\cite{subsystemFootnote}.

	Third, motivated by the use of similar analyses to understand sampling advantages in random quantum circuits~\cite{boixo2018characterizing}, we characterize the output distributions that arise when Scrooge-random states are measured in any product basis.
	We show that the statistics of the output probabilities, which are useful for device benchmarking and verification~\cite{mark2023benchmarking,shaw2024benchmarking}, are described by the so-called Wishart distribution~\cite{wishart1928generalised}, a multi-variate generalization of the Porter-Thomas distribution.
	We also show that the output distributions exhibit long-ranged conditional mutual information similar to random circuits, suggesting they might also be  hard to sample from classically.

	Finally, we apply our results to characterize the \textit{quantum} complexity of Scrooge-random states, by providing a simple lower bound on the number of bits of randomness needed for any ensemble of quantum states to replicate the moments of the Scrooge ensemble. This allows us to lower bound the evolution time needed for the Scrooge ensemble to emerge from fixed Hamiltonian dynamics, and the number of two-qubit gates needed for it to emerge in quantum circuits.

	\emph{Background.}---The Scrooge ensemble $\mathcal{S}_\rho$ is a probability distribution over pure quantum states.
	The distribution is parametrized by a fixed ``background'' density matrix $\rho$, which will allow one to encode additional structure such as charge or energy conservation.
	Specifically, $\mathcal{S}_\rho$ is defined as the unique, maximally entropic ensemble of pure states whose first moment, $\E_{\psi \sim \mathcal{S}_\rho} \dyad{\psi}$, is equal to $\rho$.
	The notion of maximum entropy is defined with respect to the accessible information for state ensembles; we refer to earlier works for details~\cite{jozsa1994lower, mark2024maximum}.
	Remarkably, an explicit construction of the Scrooge ensemble following this definition is possible~\cite{jozsa1994lower}. 
	For any $\rho$, the Scrooge ensemble is obtained by drawing pure states, $\ket{\psi_\phi} \equiv \sqrt{\rho} \ket{\phi} / \sqrt{p_\rho(\phi)}$, with probability density, $p_\rho(\phi) \equiv \bra{\phi} \rho \ket{\phi}$.
	This is the same ensemble of states that one would obtain by performing a Haar-random rank-1 measurement on the ancilla qubits of any purification of $\rho$~\footnote{Specifically, if $\rho$ is supported on a system $S$, its purification is any pure state $\ket{\Psi^{QQ'}_\rho}$ for which $\Tr_{Q'} \dyad*{\Psi^{QQ'}_\rho} = \rho^Q$.}.
	When $\rho$ is proportional to the identity matrix, the Scrooge ensemble reduces to the Haar ensemble.
	
	As for the Haar ensemble, the emergence of the Scrooge ensemble in physical settings is most naturally characterized through its moments~\cite{cotler2021emergent,mark2024maximum}.
	The $k$-th moment of the ensemble is defined as, $\chi^{(k)}_{\mathcal{S}_\rho} \equiv \E_{\psi \sim \mathcal{S}_\rho} \dyad{\psi}^{\otimes k} = \int d\phi \, p_\rho(\phi)  \dyad{\psi_\phi}^{\otimes k}$.
	The moments capture the statistics of various properties of $\psi$, such as the variances and higher cumulants of its expectation values, and the expected outcomes of multi-copy measurements. 
	
	The Scrooge ensemble and its moments have been predicted to describe an increasingly wide range of quantum many-body settings~\cite{goldstein2006distribution,goldstein2016universal,cotler2021emergent,mark2024maximum,teufel2025canonical,shaw2024benchmarking,liu2024deep,manna2025projected,sherry2025information,chang2025deep,mok2025nature,schuster2025fast}.
	Broadly speaking, these  fall into three classes.
	%
	The first class are \emph{temporal ensembles}, such as $\mathcal{E}_t = \{ e^{-iHt} \ket{\psi_0} \, | \, t \sim [0,T) \}$, for an ergodic Hamiltonian $H$, initial state $\ket{\psi_0}$ and maximum evolution time $T$. 
	Temporal ensembles describe the typical late-time behavior of pure states in chaotic many-body systems. 
	Under certain, generic conditions, they are equivalent to Scrooge ensembles, setting $\rho \equiv \int_0^T dt e^{-iHt} \dyad{\psi_0} e^{iHt}$~\footnote{To be specific, Ref.~\cite{mark2024maximum} shows that the temporal ensemble for asymptotically large $T$ is equivalent to the so-called \emph{random phase} ensemble, assuming the no-resonance condition for the Hamiltonian $H$. The moments of the random phase ensemble are computed up to small additive error in Ref.~\cite{mark2024maximum}, and produce an identical formula to our approximate expression for the moments of the Scrooge ensemble.
		Hence, the two ensembles are equivalent up to small additive error.
		We suspect that this equivalence also holds up to small relative error for observables that are sufficiently  spread out in the energy eigenbasis.}. The Scrooge ensemble, as opposed to the Haar ensemble, is required to capture this behavior due to conservation of energy under $H$~\cite{cui2025random}. 
	
	The second class are \emph{projected ensembles}, where one projectively measures a subsystem $B$ of a bipartite quantum state, $\ket{\psi_0^{AB}}$, and considers the ensemble of post-measurement pure states on the unmeasured subsystem $A$~\cite{cotler2021emergent}.
	%
	%
	The first moment of this ensemble is given by the reduced density matrix, $\rho_A \equiv \Tr_B(\dyad*{\psi_0^{AB}})$.
	The higher moments have been observed to match those of the Scrooge ensemble with respect to $\rho_A$~\cite{cotler2021emergent,mark2024maximum,shaw2024benchmarking,liu2024deep,manna2025projected,sherry2025information,chang2025deep,mok2025nature}.
	%
	
	The final class are \emph{dynamical ensembles}, where the evolution time is fixed but the rules governing evolution are randomized.
	This is analogous to the emergence of the Haar ensemble (i.e.~state and unitary designs) in local random circuits of a given depth~\cite{brandao2016local,chen2024incompressibility}.
	Building on our results, forthcoming work shows that the Scrooge ensemble emerges at short poly-logarithmic times in $U(1)$-symmetric random circuits~\cite{schuster2025fast}, mirroring the fast formation of designs in random circuits without symmetry~\cite{schuster2024random}.
	%

	\emph{Moments of Scrooge ensembles.}---The aim of this work is to characterize the universal features of Scrooge-random states in many-body quantum systems, agnostic to their physical origin.
	Accordingly, our results will apply to each  setting described above.
	To maintain physical relevance, we
	focus on
	ensembles whose \emph{low moments} approximate those of the Scrooge ensemble, but whose high moments may differ. 
	%
	%
	This mirrors standard approaches in random circuits and unitary designs~\cite{emerson2003pseudo,harrow2009random,brown2012scrambling, brandao2016local,haferkamp2022random, chen2024incompressibility,laracuente2024approximate,schuster2024random,laracuente2025quantum}. 
	
	To be specific, we consider the following notions of approximation error, following the design literature~\cite{brandao2016local,schuster2024random}:
	\vspace{-4mm}
	\begin{definition}
		An ensemble of states $\mathcal{E}$ forms an $\varepsilon$-approximate Scrooge $k$-design with relative (additive) error if its $k$-th moments  are $\varepsilon$-close to those of the Scrooge ensemble  in relative (additive) error.
	\end{definition} 
	\noindent Here, two operators $\mathcal A$ and $\mathcal B$ are close in additive error if $\lVert \mathcal A - \mathcal B \rVert_1 \leq \varepsilon$, where $\| \mathcal{X} \|_1 = \Tr \sqrt{\mathcal{X}^\dagger \mathcal{X}}$ is the trace norm.
	The operators are close in relative error if $(1-\varepsilon)\mathcal{A} \preceq \mathcal{B} \preceq (1+\varepsilon) \mathcal{A}$, where $\mathcal{X} \preceq \mathcal{Y}$ denotes that $\mathcal{Y} - \mathcal{X}$ is positive semi-definite~\cite{brandao2016local}. 
	A relative error $\varepsilon$ (which implies but is not implied by an additive error $2\varepsilon$) is sufficient to bound the expectation value of any positive $k$-copy operator up to a small multiplicative difference~\cite{schuster2024random}.
	
	Unlike the Haar ensemble, no analytic formula for the moments of the Scrooge ensemble is known.
	This has significantly limited rigorous  calculations of its properties so far.
	%
	Our main technical contribution, from which all subsequent results follow, is a simple approximation formula for the $k$-th moments of the Scrooge ensemble, with strong and precise error bounds. 
	\begin{theorem}
		\label{thm:approx}
		For any $k \leq 2^{S_\infty(\rho)/3 - 1}$. 
		The $k$-th moment of the Scrooge ensemble is approximated by
		\begin{align} \label{eq:approx}
			\chi^{(k)}_{\mathcal{S}_\rho} \approx \rho^{\otimes k} \cdot \sum_{\pi \in S_k} \pi,
		\end{align}
		up to relative error $\delta_{\rho, k}$, where
		\begin{align}
			\delta_{\rho, k} = \mathcal{O}(k^2\, 2^{-S_{\infty}(\rho)/2}).
		\end{align}
		Here, $S_\infty(\rho) = -\log(\max \textup{eig}\, \rho)$ is the min-entropy of $\rho$.
	\end{theorem}
	\noindent The summation in Eq.~\eqref{eq:approx} is over all $k!$ permutations of the $k$ copies of the Hilbert space. The min-entropy $S_\infty(\rho)$, which controls the error $\delta_{\rho, k}$ of the approximation, is the $\alpha \rightarrow \infty$ limit of the $\alpha$-R{\'e}nyi entropies $S_\alpha \coloneqq (1-\alpha)^{-1}\tr(\rho^\alpha)$, and can be lower bounded as $S_\infty \geq \frac{1}{2}S_2 = - \frac{1}{2}\log \tr(\rho^2)$. In typical many-body settings where $\rho$ has finite entropy density, we expect $S_\infty(\rho) \propto n$, in which case Eq.~\eqref{eq:approx} holds up to an exponentially small error in the system size.
	
	
	Our approximate formula can be equivalently obtained as follows.
	One considers a Haar-random state $\ket{\phi}$, and deforms it via $\ket{\phi} \rightarrow \ket{\tilde{\psi}_\phi} \coloneqq \sqrt{\rho}\ket{\phi}$, without properly re-normalizing afterwards.
	The $k$-th moments of this ensemble are precisely  equal to Eq.~(\ref{eq:approx}), up to a near-unity constant prefactor. 
	The idea of using this ensemble to approximate the Scrooge moments was anticipated in~\cite{mark2024maximum}, where a small additive error for the second moment $k = 2$ was argued for.
	Our Theorem~\ref{thm:approx} strengthens this to  a rigorous relative-error bound that extends to higher moments $k>2$. These improvements will be essential for 
	most of our ensuing results.
	
	Our proof of Theorem~\ref{thm:approx} combines a careful analysis of the relevant integrals over $\phi$ with strong concentration bounds on the fluctuations of the normalizing factor $p_\rho(\phi) = \braket{\tilde{\psi}_\phi|\tilde{\psi}_\phi}$. A key step is to leverage Schur convexity to show that, for any $\rho$, $p_\rho(\phi)$ concentrates at least as sharply about its mean as the case where $\rho$ has a flat spectrum of rank $\lfloor2^{S_\infty(\rho)}\rfloor$. By simplifying the spectrum of $\rho$ in this way, many of the requisite calculations become tractable.
	We refer to the Supplemental Material (SM) for further details and our full technical proof~\cite{supp}.
	%

	In the remainder of our work, we describe several important physical properties of ensembles of many-body quantum states that form approximate Scrooge $k$-designs, each of which we derive using Theorem \ref{thm:approx}.
	
	\textit{Concentration of reduced states.---} 
	We begin by considering the reduced density matrices of Scrooge-random states.
	In many scenarios, one is interested in observables that are localized to some subregion $A$. The statistics of such quantities are fully captured by the moments of the marginal states, $\chi_{\mathcal{S}_\rho,A}^{(k)} \coloneqq \E_{\psi \sim \mathcal{S}_\rho} [\tr_B(\dyad{\psi})]^{\otimes k}$, which result from tracing out the complementary subsystem $B$ from the full moments $\chi^{(k)}_{\mathcal{S}_\rho}$.
	
	Unlike the full moments,  we find that reduced moments of Scrooge-random states concentrate extremely sharply about their values in $\rho$. 
	To be specific, we prove that the reduced  moments are approximated by,
	\begin{align}
		\chi_{\mathcal{S}_\rho,A}^{(k)} \approx \tr_B(\rho)^{\otimes k},
		\label{eq:subsystem approx}
	\end{align}
	up to a relative error that is---for typical many-body $\rho$ (see below)---exponentially small in the size of $B$, i.e.~the number of qubits traced out. The moments on the right  pertain to an ensemble in which each random state $\dyad{\psi}$ is replaced by the fixed state $\rho$. Hence, Eq.~(\ref{eq:subsystem approx}) shows that, much like the Haar ensemble~\cite{HaydenLeungWinter2006}, instance-to-instance fluctuations of Scrooge-random states can only appear in fully non-local observables.
	
	In the SM~\cite{supp}, we derive fully general relative-error bounds on the accuracy of the approximation Eq.~\eqref{eq:subsystem approx}. As an important special case, for any set of $k$ positive semi-definite observables, $(O_1^A, \ldots, O_k^A)$,  supported on $A$, the moments $\E_\psi\big[\prod_{r=1}^k\braket{\psi|O^A_r|\psi}\big]$ are equal to $\prod_{r=1}^k\text{tr}(\rho O^A_r)$ up to relative error $\mathcal{O}(k^2 2^{-S^*_\infty(B)})$. Here, $S_\infty^*(B)$ is the smallest min-entropy over all states on $B$ that can be obtained by performing a projective measurement on $A$, starting from $\rho$. In typical many-body states $\rho$ with finite entropy density, we expect $S_\infty^*(B)$ to scale linearly with $|B|$. 
	Thus, the moments collapse exponentially quickly in the number of traced out qubits.

	\emph{Output distributions.}---To further explore the implications of our approximate formulas for the moments [Eqs.~(\ref{eq:approx}) and~(\ref{eq:subsystem approx})], we consider the distribution of output when a Scrooge-random state $\ket{\psi}$ is measured in the computational basis.
	This produces a bitstring $x\in\{0,1\}^{\otimes n}$ with probability $p^\psi(x) = |\langle x | \psi \rangle|^2$. 
	Such distributions have been of intense recent interest in quantum sampling advantages and benchmarking protocols~\cite{boixo2018characterizing,arute2019quantum,shaw2024benchmarking,mark2023benchmarking}.
	
	We begin by characterizing fundamental properties of the output distribution, and then discuss practical implications.
	%
	For each bitstring $x$, the probability $p^\psi(x)$ is a random variable over the Scrooge ensemble, with a mean value $\E_\psi [p^\psi(x)] = p^\rho(x) \coloneqq \braket{x|\rho|x}$. Remarkably, this variable does not concentrate, and instead features strong instance-to-instance fluctuations. We can characterize this through moments of $p^\psi(x)$, which can be computed easily following Theorem \ref{thm:approx}. For any ensemble $\mathcal{E}$ that forms an $\epsilon$-approximate Scrooge $k$-design with relative error, we have
	\begin{align}
		\E_{\psi \sim \mathcal{E}}p^\psi(x)^k & \approx k! \, p^\rho(x)^k,
		\label{eq:porterthomas}
	\end{align}
	up to relative error $\mathcal{O}(\epsilon + k^2 2^{-S_\infty(\rho)/2})$.
	The factor of $k!$ arises from the sum over permutations $\pi \in S_k$. 
	This matches the moments of a re-scaled Porter-Thomas distribution, $\text{Pr}(p) = \frac{1}{p^\rho(x)} e^{-p/p^\rho(x)}$, up to small relative error, as conjectured and numerically observed in~\cite{mark2024maximum}. 
	From Eq.~\eqref{eq:porterthomas}, we can also prove that the output distribution $p^\psi(x)$ is on average far from $p^\rho(x)$ in total variation distance (TVD), $\E_{\psi \sim \mathcal{E}}\sum_x|p^\psi(x) - p^\rho(x)| = \Theta(1 - \epsilon)$, whenever $\mathcal{E}$ is an $\epsilon$-approximate Scrooge 4-design~\cite{supp}. 
	%
	
	Intriguingly, we find that the joint moments of the probabilities of two bitstrings, $p^\psi(x)$ and $p^\psi(x’)$, do \emph{not} follow a product of independent Porter-Thomas distributions, and instead feature correlations characterized by the so-called Wishart distribution \cite{wishart1928generalised}. In the SM~\cite{supp}, we describe how these correlations impact the sensitivity of the output distribution to read-out noise.
	
	Whereas the global output distribution exhibits strong fluctuations, Eq.~\eqref{eq:subsystem approx} implies that the marginals of the output distribution, $p^\psi_A(x_A) = \braket{x_A|\tr_B(\dyad{\psi}) |x_A}$ for a subregion $A$,  exhibit very small fluctuations. Specifically,
	\begin{align}
		\E_{\psi \sim \mathcal{E}} p^\psi_A(x_A)^k \approx p^\rho_A(x_A)^k,
		\label{eq:marginal prob conc}
	\end{align}
	up to relative error $\mathcal{O}(\epsilon + k^22^{-S_\infty^*(B)})$,
	where $B$ is the complement of $A$. 
	As a result, we find that the TVD between $p^\psi_A(x_A)$ and the (non-random) background distribution $p^\rho_A(x_A)$ is exponentially small in the size of $B$, assuming $S^*_\infty(B) = \Theta(|B|)$ as before~\cite{supp}.

	\textit{Conditional correlations and hardness of sampling.---}The emergence of Porter-Thomas fluctuations in the output probabilities of Scrooge random states is reminiscent of random quantum circuits. This suggests that Scrooge-random quantum states may be similarly hard to sample from classically. While a full investigation of this question lies beyond the scope of our work, we provide several observations that support this intuition.
	
	First, we show that a particular widely-considered set of classical algorithms~\cite{napp2022efficient,lee2025classical,zhang2025classically, wei2025measurement} fail to simulate sampling from Scrooge-random states.
	These `patching' algorithms proceed by sampling from one subregion of a state at a time, and, for each subregion, conditioning the sampling only on previous outcomes within nearby subregions.
	This partial conditioning makes the algorithm efficient, but is valid only if the conditional mutual information (CMI) of the output distribution is small between faraway subregions.
	In the supplement~\cite{supp}, we show that for any geometry, the average CMI of states drawn from approximate Scrooge 4-designs is nonzero and at least of order unity, thereby invalidating such approaches. Similarly, algorithms that use tensor network ans{\"a}tze to represent the intermediate conditional states (e.g.~sideways-evolving block decimation \cite{napp2022efficient, wei2025measurement}) are also likely to fail due to the large bipartite entanglement in Scrooge-random states \cite{goldstein2016universal}.
	
	While Porter-Thomas fluctuations and long-ranged CMI are associated with hardness of sampling, we also show that these phenomena quickly vanish under small amounts of noise.
	We consider local depolarizing noise of strength $\gamma$ applied to each qubit before measurement.
	This modifies the output distribution to $\tilde{p}^\psi(x) = \braket{x|\mathcal{D}_\gamma(\dyad{\psi})|x}$.
	Using Eq.~\eqref{eq:marginal prob conc}, we prove that the TVD between $\tilde{p}^\psi(x)$
	and the output distribution of the noisy background state, $\tilde{p}^\rho(x) = \braket{x|\mathcal{D}_\gamma[\rho]|x}$, decays exponentially in $\gamma n$ for typical many-body $\rho$~\cite{supp}. Thus, any algorithm that simulates sampling from $\rho$ (which is generally highly mixed) can be used to sample from a very weakly noisy Scrooge-random state $\tilde{p}^\psi(x)$. This reflects a similar phenomenon in random circuits, where small amounts of noise lead to a fast collapse of the distribution to uniform \cite{deshpande2022tight, dalzell2021random}.
	This also supports the use of Porter-Thomas fluctuations as a highly sensitive diagnostic for certifying the absence of noise in a quantum device, as in Ref.~\cite{mark2023benchmarking,shaw2024benchmarking}.
	
	\emph{Constraints on the formation and complexity of Scrooge $k$-designs.}---We now turn to our final application of our approximate formula [Eq.~(\ref{eq:approx})], to provide a lower bound on the number of bits of randomness $m_{\rm bits}$ required to generate any approximate Scrooge $k$-design, i.e.~the logarithm of the number of distinct states in the ensemble. 
	In the supplement, we prove that
	the number of bits of randomness in any state ensemble $\mathcal{E}$ that forms an $\epsilon$-approximate Scrooge $k$-design with additive error must be at least
	\begin{align}
		m_{\rm bits}[\mathcal{E}] \geq k (S_\infty(\rho) - \log k) - \log(1 - \epsilon - 2\delta_{\rho,k}),
		\label{eq:bits}
	\end{align}
	where $\delta_{\rho,k}$ is the approximation error in Theorem \ref{thm:approx}.
	\noindent We emphasize that the above bound must be satisfied even if one only demands \textit{additive} error designs.
	This mirrors lower bounds for Haar state $k$-designs, which require $m_{\text{bits}} = \Omega(nk)$ bits of randomness~\cite{gross2007evenly,brandao2016local}.

	As a simple application of this bound, we consider the time $T$ required for the temporal ensemble to form an approximate Scrooge $k$-design. 
	As previously mentioned, the temporal ensemble has been widely used as a model of chaotic many-body Hamiltonian dynamics, and discussed in relation to Scrooge ensembles. Evolving a state under a fixed Hamiltonian for a random time $t \sim \text{Unif}[0,T]$ generates an ensemble with effectively $m_{\rm bits} = \mathcal{O}(\log T)$ bits of randomness \cite{cui2025random}.  
	Interestingly, as a result of \eqref{eq:bits}, we find that the temporal ensemble cannot form a Scrooge $k$-design unless $T = \Omega(e^{kS_\infty(\rho)}(1-\epsilon)^2/n)$.
	This is exponentially large in $nk$ for typical $\rho$, where $S_\infty(\rho) = \Theta(n)$.
	
	Finally, we use Eq.~\eqref{eq:bits} to prove that states drawn from Scrooge $k$-designs have high complexity, scaling with $nk$ for typical $\rho$. A pure state $\ket{\psi}$ has a $\delta$-robust complexity $\mathcal{C}_\delta(\ket{\psi})$ of at least $r$ if all quantum circuits made up of $r$ two qubit gates fail to prepare $\ket{\psi}$ to an accuracy $\delta$. It has been conjectured that the state of a time-evolving many-body system should generically exhibit growth of complexity over exponentially long time periods \cite{susskind2016computational, brown2016complexity}, well beyond the thermalization time. In settings without conservation laws, the formation of spherical designs has been shown to be a powerful proxy for lower bounding complexity growth \cite{brandao2016local, roberts2017chaos, brandao2019models}; here we prove a similar result for Scrooge-designs.
	
	We show that if $\ket{\psi}$ is drawn from a $\epsilon$-approximate Scrooge $k$-design with \textit{additive} error, then its $\delta$-robust complexity is at least
	\begin{align}
		\mathcal{C}_\delta(\ket{\psi}) = \Omega\left(\frac{k(S_\infty(\rho)-\log k)}{\log[kS_\infty(\rho)]} \right)
	\end{align}
	with high probability (made precise in the supplement \cite{supp}). In the typical case $S_\infty(\rho) = \Theta(n)$, this gives a quasi-linear scaling with $nk$. Thus, the formation of Scrooge $k$-designs with increasing $k$ can be used to witness the growth of complexity.

	\emph{Outlook.}---Our results provide an analytic foundation for future studies of the emergence and applications of Scrooge ensembles in many-body quantum settings. Multiple important questions remain open. What source of randomness seeds the Scrooge ensemble in  physical Hamiltonian time-evolutions? Our lower bounds on temporal ensembles indicate  it cannot be the evolution time itself, in tension with common assumptions. Can we more sharply characterize when one expects the Scrooge ensemble emerge? Our results provide many signatures to test for its behavior. 
	Can the precise characterization of Scrooge-random states provided by our work be applied to quantum learning applications, analogous to recent progress using Haar-random ensembles \cite{huang2020predicting,aharonov2022quantum,huang2021quantum,chen2021exponential}? 
	And finally, the permutation structure approximate formula for the Scrooge ensemble moments bears close resemblance to symmetry-breaking descriptions of measurement-induced phase transitions \cite{bao2020theory, jian2020measurement}.
	Can this perspective be used to understand the emergence of Scrooge-random states, and conversely, can the Scrooge ensemble capture these behaviours?

	\textit{Acknowledgements}---We are grateful to David Gosset, Wai-Keong Mok, Yinchen Liu, John Preskill, and  Shreya Vardhan for discussions on related work. M.M.~acknowledges support from Trinity College, Cambridge.
	T.S. acknowledges support from the U.S. Department of Energy, Office of Science, National Quantum Information Science Research Centers, Quantum Systems Accelerator.
	The Institute for Quantum Information and Matter is an NSF Physics Frontiers Center.

	
	\let\oldaddcontentsline\addcontentsline
	\renewcommand{\addcontentsline}[3]{}
	\bibliography{refs}
	\let\addcontentsline\oldaddcontentsline

	
	\onecolumngrid
	\newpage
	
	{\centering
		\large\bfseries
		Supplementary Material: The Scrooge ensemble in many-body quantum systems
		\par}

	\tableofcontents

	\section{Simple approximate formula for the moments of the Scrooge ensemble} \label{app:approximate}
	
	In this section, we present the full details and intuition behind our simple approximate formula for the moments of the Scrooge ensemble, which was stated as Theorem \ref{thm:approx} in the main text. We begin by introducing representations for the exact and approximate Scrooge moments as integrals over many-body Hilbert space. The two expressions are are related to one another by a scalar weighting factor $r_\phi$, and we prove that this factor concentrates very strongly around unity. 
	
	\subsection{Integral representations of Scrooge moments}
	
	The first given definition of the Scrooge ensemble was as a `$\rho$-distortion' of the Haar ensemble \cite{jozsa1994lower} (also referred to as an adjusted projected ensemble in Refs.~\cite{goldstein2006distribution, goldstein2016universal}), which has the following explicit construction. Starting from the Haar measure $\dif \mu_{\rm Haar}(\phi)$ over the space of normalized $d$-dimensional wavefunctions, the $\rho$-distortion involves modifying the states as $\ket{\phi} \mapsto \ket{\psi(\phi)} \coloneqq \frac{\sqrt{\rho}\ket{\phi}}{\|\sqrt{\rho}\ket{\phi}\|}$, and reweighting the probability measure by a scalar factor
	\begin{align}
		r_\phi \coloneqq d \braket{\phi|\rho|\phi}.
	\end{align}
	Since the average density matrix for the Haar ensemble is maximally mixed, we have $\int \dif \mu_{\rm Haar}(\phi) r_\phi = 1$. Thus the resulting measure for the Scrooge ensemble $\dif \nu_\rho(\psi)$ is defined over wavefunctions of unit norm, and is itself normalized, $\int \dif \nu_\rho(\psi) = 1$. Given a function $f(\ket{\psi})$, its average over the Scrooge ensemble can be written
	\begin{align}
		\mathbb{E}_{\psi \sim \mathcal{S}_\rho}\big[f(\ket{\psi})\big] \equiv \int \dif \nu_\rho(\psi) f(\ket{\psi}) = \int \dif \mu_{\rm Haar}(\phi) r_\phi  f\left( \frac{\sqrt{d \rho}\ket{\phi}}{\sqrt{r_\phi}} \right)
		\label{eq:Scrooge expectation}
	\end{align}
	In particular, the the matrix of $k$th moments becomes
	\begin{align}
		\chi^{(k)}_{\mathcal{S}_\rho} = \int \dif \mu_{\rm Haar}(\phi) r_\phi^{1-k} \Big( \sqrt{d\rho}\ket{\phi}\bra{\phi}\sqrt{d\rho}\Big)^{\otimes k}
		\label{eq:scrooge moment exact}
	\end{align}
	The above expression is hard to evaluate in general due to the factor $r_\phi^{1-k}$, which precludes any direct usage of the well-known positive integer moments of the Haar ensemble. To obtain the approximate moments, we can instead replace $r_\phi$ with its mean value of unity, yielding the unnormalized moments
	\begin{align}
		\chi^{(k)}_{\mathcal{S}_\rho, {\rm Un}} = \int \dif \mu_{\rm Haar}(\phi)  \Big( \sqrt{d\rho}\ket{\phi}\bra{\phi}\sqrt{d\rho}\Big)^{\otimes k} = (\sqrt{d \rho})^{\otimes k}\chi^{(k)}_{\rm Haar} (\sqrt{d \rho})^{\otimes k}.
		\label{eq:chi approx def 1}
	\end{align}
	Using well-known expressions for the Haar moments, this can be evaluated in closed form
	\begin{align}
		\chi^{(k)}_{\mathcal{S}_\rho, {\rm Un}} = (d\rho)^{\otimes k} \frac{k!}{d^{(k)}}\Pi_{\rm sym}^{(k)} = \frac{d^k}{d^{(k)}}\rho^{\otimes k}\sum_{\pi \in S_k}\pi \approx \rho^{\otimes k}\sum_{\pi \in S_k}\pi \eqqcolon \chi^{(k)}_{\mathcal{S}_\rho, {\rm Appr}}
		\label{eq:chi approx def 2}
	\end{align}
	where $d^{(k)} = d(d+1)\cdots(d+k-1)$ is the rising factorial, and the relative error of the approximation is $1 - \frac{d^k}{d^{(k)}} = 1 - \prod_{l=1}^{k-1}(1 + l/d) \leq 1 - e^{-k(k-1)/2d} = \mathcal{O}(k^2/d)$. Here, $\Pi_{\rm sym}^{(k)}$ is the projector onto the $k$-copy symmetric subspace $\mathcal{H}^{(k)}_{\rm sym} = \text{span}(\{\ket{\phi^{\otimes k}}: \phi \in \mathcal{H}\})$, and in going from Eq.~\eqref{eq:chi approx def 1} to Eq.~\eqref{eq:chi approx def 2},
	we have used the fact that $\rho^{\otimes k}$ commutes with $\Pi_{\rm sym}^{(k)}$. Thus, we can view the approximate moments, defined on the right hand side of \eqref{eq:chi approx def 2} and appearing in Eq.~\eqref{eq:approx} of the main text, as being close to those induced by a distribution over \textit{unnormalized} wavefunctions that is obtained by multiplying Haar-random states $\ket{\phi}$ by the operator $\sqrt{d \rho}$ (see also Ref.~\cite{mark2024maximum}).\\
	
	We already have that $\mathcal{E}^{(k)}_{\rm Appr}$ and $\mathcal{E}^{(k)}_{\rm Un}$ are approximately equal to small \textit{relative} error $\mathcal{O}(k^2/d)$. Theorem \ref{thm:approx} will then be proved if we can show that $\mathcal{E}^{(k)}_{\rm Un}$ and $\mathcal{E}^{(k)}_{\rm Scr}$ are also close in relative error. To establish this, we will show that $r_\phi$ concentrates very strongly around its mean over the Haar ensemble, such that the integrals in \eqref{eq:scrooge moment exact} and \eqref{eq:chi approx def 1} become approximately equal. By the rotational invariance of the Haar ensemble, the induced distribution of $r_\phi$ depends only on the eigenvalues of $\rho$, which we write as $\lambda_i$. 
	We fully characterize the distribution as follows.
	\begin{lemma}
		\label{lem:RDist}
		For a $d$-dimensional density matrix $\rho$, the distribution of $r = d \braket{\phi|\rho|\phi}$ induced by the Haar distribution over states $\ket{\phi}$ has a probability density function
		\begin{align}
			p(r) = S_{d-1}(r; \mu_1, \ldots, \mu_d)
			\label{eq:RPDF}
		\end{align}
		where $\mu_i = d \lambda_i$, with $\lambda_i$ the eigenvalues of $\rho$ in ascending order. $S_{d-1}(x;x_1, \ldots, x_n)$ is the fundamental spline function with $d$ knots, defined Ref.~\cite{curry1966polya}.
	\end{lemma}
	For a set of (possibly degenerate) `knots' $\{x_i\}_{i=1}^d$, sorted in ascending order $x_i \leq x_{i+1}$, the fundamental spline function $S_{d-1}(x;x_1, \ldots, x_d)$ is a non-negative piecewise polynomial of $x$, which is non-analytic only at the locations of the knots, and is zero outside the region $[x_1, x_d]$. It plays an important role in polynomial interpolation, and is studied extensively in Ref.~\cite{curry1966polya}. (We note that a relationship between the Scrooge ensemble and polynomial interpolation was also noted in Appendix A of Ref.~\cite{jozsa1994lower}.) Writing $r_\phi = \sum_i \mu_i \xi_i$, where $\xi_i = |\braket{i|\phi}|^2$ are squared overlaps of $\ket{\phi}$ with eigenvectors $\ket{i}$ of $\rho$, Lemma \ref{lem:RDist} follows quickly from a particular geometric representation of the spline function (Theorem 2 of Ref.~\cite{curry1966polya}), along with Sykora's theorem (Ref.~\cite{sykora1974quantum}, Sec.~3), which states that when $\ket{\phi}$ is Haar-random, the joint distribution of $\xi_i$ is uniform over the probability simplex $\{(\xi_1, \ldots, \xi_d):\xi_i \geq 0, \sum_i \xi_i = 1\}$.
	
	The key implication of Lemma \ref{lem:RDist} is the following characterization of expectation values of functions of $r$, which follows from Eq.~(1.5) of Ref.~\cite{curry1966polya}. For a $(d-1)$-differentiable function $f(r)$, we have
	\begin{align}
		\mathbb{E}_{\phi \sim \text{Haar}}[f^{(d-1)}(r_\phi)] = (d-1)! f[\mu_1, \ldots, \mu_d]
		\label{eq:MomentDivDiff}
	\end{align}
	Here, $f^{(d-1)}$ is the $(d-1)$th derivative of $f$, while $f[x_1, \ldots, x_d]$ is the \textit{divided difference}, defined recursively as
	\begin{align}
		f[x_1, \ldots, x_{n+1}] = \frac{f[x_2, \ldots, x_{n+1}] - f[x_1, \ldots, x_n]}{x_{n+1} - x_0}.
	\end{align}
	with $f[x] = f(x)$. As can be verified from the above definition, divided differences are invariant under permutations of the arguments.
	
	\subsection{Concentration properties of $r_\phi$}
	
	Using the above characterization of the distribution of $r_\phi$, we now prove several bounds on the moments of $r_\phi$, which imply very strong concentration of $r_\phi$ around unity.
	\begin{lemma}
		\label{lem:RConc}
		Let $\rho$ be a $d$-dimensional density matrix $\rho$, and let $m$ be any integer such that $m \times \max \textup{eig } \rho \leq 1$. Then for any $q \in [0,m)$, the raw and reciprocal moments of $r_\phi = d\braket{\phi|\rho|\phi}$ for Haar-distributed states $\ket{\phi}$ satisfy
		\begin{align}
			\mathbb{E}_{\phi \sim {\rm Haar}}[r^s_\phi] &\leq \exp\left( \frac{s(s+1)}{2m}\right), & s &\geq 1, \label{eq:RRawMomentBound} \\
			\mathbb{E}_{\phi \sim \textup{Haar}}\big[ r_\phi^{-q}\big] &\leq \exp\left( \frac{q^2}{2(m-q)}\right),& 0 &\leq q < m,
			\label{eq:RReciprocalMomentBound}
		\end{align}
		while the central moments of order $t \geq 1$ satisfy
		\begin{align}
			\mathbb{E}_{\phi \sim \textup{Haar}}\big[|r_\phi - 1|^{t}\big] &\leq 4\frac{\Gamma(t+1)}{(m/4)^{t/2}},\quad\quad t \geq 1.
			\label{eq:RCentralMomentBound}
		\end{align}
	\end{lemma}
	Eq.~\eqref{eq:RReciprocalMomentBound} will be used to control the probability that $r_\phi$ becomes dangerously close to zero, while Eq.~\eqref{eq:RCentralMomentBound} shows that deviations away from $r = 1$ are have a sub-exponential distribution. We note that certain one-sided concentration bounds for $r_\phi$ were recently shown in Ref.~\cite{teufel2025canonical} (Lemma 2 therein), but given our additional need to bound the reciprocal moments \eqref{eq:RReciprocalMomentBound}, we will use a different method in the following.\\
	
	\textit{Proof of Lemma \ref{lem:RConc}.---}As noted above and made manifest by Eq.~\eqref{eq:MomentDivDiff}, the moments appearing on the left hand sides of (\ref{eq:RReciprocalMomentBound}, \ref{eq:RCentralMomentBound}) depend on $\rho$ only through its eigenvalues $\lambda_i$. We will use the following result from Ref.~\cite{Farwig1985}. If $f^{(d-1)}(r)$ is a convex function, then the divided difference $f[\mu_1, \ldots, \mu_d]$ is a convex function of the vector $\vec{\mu} = (\mu_1, \ldots, \mu_d)$. Moreover, by the invariance of divided differences under permuting arguments, we can also infer that $f[\mu_1, \ldots, \mu_d]$ is Schur convex in $\vec{\mu}$. Since $r^s$, $r^{-q}$, and $|r-1|^t$ are all convex in the specified range of $s,q$ and $t$, Eq.~\eqref{eq:MomentDivDiff} implies that the moments in question [the left hand sides of Eqs.~(\ref{eq:RRawMomentBound}, \ref{eq:RReciprocalMomentBound}, \ref{eq:RCentralMomentBound})] are Schur convex functions of the eigenvalues of $\rho$.
	
	Recall that Schur convex functions $g(\vec{x})$ satisfy the property that for two vectors $\vec{x}$, $\vec{y}$ such that $\vec{x}$ majorizes $\vec{y}$, we have $g(\vec{x}) \geq g(\vec{y})$. Majorization of $\vec{y}$ by $\vec{x}$, denoted $\vec{y} \prec \vec{x}$, means that the sum of the $k$th largest values of $\vec{y}$ are no greater than the sum of the $k$th largest values of $\vec{x}$, and also that $\sum_{i=1}^d x_i = \sum_{i=1}^d y_i$. One can verify that the spectrum of $\rho$ is majorized by the following spectrum
	\begin{align}
		\vec{\lambda} &\preceq \left( \underbrace{\frac{1}{m}, \ldots, \frac{1}{m}}_{m\text{ times}}, 0, \ldots, 0 \right) & \text{where }m = \lfloor 2^{S_\infty(\rho)}\rfloor = \max \text{eig }\rho 
	\end{align}
	Thus, for a given density matrix $\rho$, the moments we are considering are no larger than those when one sets $\rho = \Pi_m/m$, where $\Pi_m$ is a rank-$m$ projector. In that case, one can straightforwardly show that the rescaled random variable $u \coloneqq (m/d)\times r_\phi$ has the beta distribution $u \sim \mathcal{B}(m,d-m)$ with probability distribution function $p(u) = \frac{u^{m-1}(1-u)^{d-m-1}}{B(m,d-m)}$, where $B(\alpha,\beta) = \frac{\Gamma(\alpha)\Gamma(\beta)}{\Gamma(\alpha+\beta)}$ is the beta function. The probability density function can accordingly be written as
	\begin{align}
		p(r)|_{\rho = \Pi_m/m} &= \frac{m^m}{d^m B(m,d-m)} r^{m-1}(1-mr/d)^{d-m-1} & r \in [0,d/m]
	\end{align}
	where $B(\alpha,\beta) = \frac{\Gamma(\alpha)\Gamma(\beta)}{\Gamma(\alpha+\beta)}$ is the beta function. The raw moments can be computed exactly
	\begin{align}
		\mathbb{E}[r_\phi^{s}]|_{\rho = \Pi_m/m}   &= \frac{\Gamma(d+1)d^s}{\Gamma(d+1+s)} \cdot \frac{\Gamma(m+1+s)}{\Gamma(m+1)m^s}.  
	\end{align}
	For integer $s$, the identity $\Gamma(z+1) = z\Gamma(z)$ can be used to show that the first fraction is always less than unity. As for the second fraction, for $s \geq 0$ this equals $\prod_{i=1}^s(1+\frac{i}{m}) \leq e^{s(s+1)/2m}$, using the inequality $(1+x) \leq e^x$.  Similarly, for $s = -q <0$, the fraction becomes $\prod_{i=1}^{q-1} (1 + \frac{i}{m-i}) \leq \exp(\sum_{i=1}^{q-1}\ln(1 + \frac{i}{m-q})) \leq \exp(q^2/2(m-q))$. Both these bounds can also be shown to hold for non-integer $s$ using Gautschi's inequalities. Using the Schur convexity argument above, this implies Eq.~\eqref{eq:RReciprocalMomentBound}.

	To get an upper bound on the central moments, we use the following sub-gamma concentration inequality for the beta distribution, proved in Ref.~\cite{skorski2023}.
	\begin{align}
		\mathbb{P}_{u \sim \mathcal{B}(m,d-m)}(|u - \bar{u}| > \epsilon) \leq 2\exp\left(- \frac{\epsilon^2}{2(\sigma^2_u+ \frac{c \epsilon}{3})}\right),
	\end{align}
	where $\sigma^2_u = \frac{m(d-m)}{d^2(d+1)}$ is the variance of $u$ and $c = \frac{d-2m}{2d(d+2)}$ is the `scale parameter', which governs the crossover from Gaussian- to exponential-type concentration \cite{boucheron2013}. Applied to the rescaled variable $r_\phi = du/m$, we get
	\begin{align}
		\mathbb{E}[|r_\phi-1|^{t}] &= t \int_0^\infty\dif x\,x^{t-1} \mathbb{P}\big(|r_\phi-1| > x\big) \nonumber\\
		&\leq 2t \int_0^\infty \dif x\, x^{t-1} \exp\left(-\frac{mx^2}{2(1+x/6)} \right)\nonumber\\
		&= 2t m^{-\ell}\int_0^\infty \dif y\, y^{t -1}\exp\left(-\frac{1}{2}\frac{y^2}{1+\frac{y}{6\sqrt{m}}} \right) \nonumber\\
		&\leq 2t m^{-\ell}\int_0^\infty \dif y\, y^{t -1}(e^{-y^2/4}+e^{-3\sqrt{m}y/2}) = \frac{2\Gamma(t/2 + 1)}{(m/4)^{t/2}} + \frac{2 \Gamma(t+1)}{(3m/2)^{t}} \leq 4 \frac{\Gamma(t+1)}{(m/4)^{t/2}},
	\end{align}
	which gives us Eq.~\eqref{eq:RCentralMomentBound}. In going from the second to the third line, we have used the inequality $\exp(-\frac{1}{a+b}) \leq \exp(-\frac{1}{2\max(a,b)}) \leq e^{-1/2a} + e^{-1/2b}$, valid for $a,b > 0$. 
	
	\subsection{Relating the exact and approximate moments}
	
	Our strategy now will be to show that the conjugated moments
	\begin{align}
		\omega^{(k)}_{\mathcal{S}_\rho} &\coloneqq \big((d\rho)^{-1/2}\big)^{\otimes k} \chi^{(k)}_{\mathcal{S}_\rho}\big((d\rho)^{-1/2}\big)^{\otimes k} = \int \dif \mu_{\rm Haar}(\phi)\, r_\phi^{1-k} \big(\ket{\phi}\bra{\phi}\big)^{\otimes k} \\
		\omega^{(k)}_{\mathcal{S}_\rho, {\rm Un}} &\coloneqq \big((d\rho)^{-1/2}\big)^{\otimes k} \chi^{(k)}_{\mathcal{S}_\rho, {\rm Un}}\big((d\rho)^{-1/2}\big)^{\otimes k} = \int \dif \mu_{\rm Haar}(\phi)\,  \big(\ket{\phi}\bra{\phi}\big)^{\otimes k} = \chi^{(k)}_{\rm Haar}
		\label{eq:conj mom def}
	\end{align}
	are close in relative error, i.e.~that
	\begin{align}
		(1-\delta)\omega^{(k)}_{\mathcal{S}_\rho, {\rm Un}} \preceq \omega^{(k)}_{\mathcal{S}_\rho} \preceq (1+\delta)\omega^{(k)}_{\mathcal{S}_\rho, {\rm Un}}
		\label{eq:RelErrorConjugated}
	\end{align}
	for some $\delta \geq 0$. Since conjugation by Hermitian operators preserves semidefinite ordering, this will in turn prove that $\chi^{(k)}_{\mathcal{S}_\rho}$ and $\chi^{(k)}_{\mathcal{S}_\rho, {\rm Un}}$ are close in relative error. When combined with the small relative-error approximation \eqref{eq:chi approx def 2}, this will be sufficient to prove Theorem \ref{thm:approx}.
	
	We take an arbitrary state $\ket{x} \in \mathcal{H}_{\rm sym}^{(k)}$ supported in the $k$-copy symmetric subspace, and define $\Pi_\phi \coloneqq \ket{\phi}\bra{\phi}^{\otimes k}$ for convenience. Then,
	\begin{align}
		\braket{x|\omega^{(k)}_{\mathcal{S}_\rho, {\rm Un}}-\omega^{(k)}_{\mathcal{S}_\rho}|x} &= \mathbb{E}_{\phi \sim \text{Haar}}\Big[r^{1-k}_\phi(r^{k-1}_\phi-1)\braket{x|\Pi_\phi|x}\Big]
	\end{align}
	From the identity $x^{n} - 1 = (x-1)(1 + x + \cdots + x^{n-1})$, we can break up this expression into separate terms and apply the triangle inequality
	\begin{align}
		\label{eq:ScrDiff}
		|\braket{x|\omega^{(k)}_{\mathcal{S}_\rho, {\rm Un}}-\omega^{(k)}_{\mathcal{S}_\rho}|x}| &=  \left|\sum_{a=0}^{k-2}\mathbb{E}_{\phi \sim \text{Haar}}\Big[(r_\phi-1)r_\phi^{1-k+a}\braket{x|\Pi_\phi|x} \Big]\right| \nonumber\\
		&\leq \sum_{a=0}^{k-2}\left|\mathbb{E}_{\phi \sim \text{Haar}}\Big[|r_\phi-1|r_\phi^{1-k+a}\braket{x|\Pi_\phi|x} \Big]\right| \nonumber\\
		&\leq \sum_{a=0}^{k-2} \|r_\phi-1\|_{L^p} \|r_\phi^{1-k+a}\|_{L^q} \|\braket{x|\Pi_\phi|x}\|_{L^s} &\text{for any }p,q,s \geq 1\text{ s.t. }\frac{1}{p} + \frac{1}{q} + \frac{1}{s} \leq 1,
	\end{align}
	where the $L^p$ norms over the Haar measure are defined as $\|f(\phi)\|_{L^p} \coloneqq (\int \dif\mu_{\rm Haar}(\phi) |f(\phi)|^p)^{1/p}$, and H{\"o}lder's inequality has been used in the last step. We now demand that $s \in [1,2]$. For such fractional values of $s$, a second usage of H{\"o}lder's inequality gives
	\begin{align}
		\|\braket{x|\Pi_\phi|x}\|_{L^s} &\leq \|\braket{x|\Pi_\phi|x}\|_{L^1}^{2/s - 1} \|\braket{x|\Pi_\phi|x}\|_{L^2}^{2 - 2/s}  & \forall s \in [1,2].
	\end{align}
	Since $\braket{x|\Pi_\phi|x}$ is a non-negative variable, the integer moments appearing in the above expression can be computed.
	\begin{align}
		\|\braket{x|\Pi_\phi|x}\|_{L^1} &= \mathbbm{E}_{\phi \sim \text{Haar}}\Big[ \braket{x|\Pi_\phi|x}\Big] = \frac{(k)!}{d^{(k)}} \braket{x|\Pi_{\rm sym}^{(k)}| x} = \frac{(k)!}{d^{(k)}} \\
		\|\braket{x|\Pi_\phi|x}\|_{L^2}^{2} &= \mathbbm{E}_{\phi \sim \text{Haar}}\Big[ \braket{x|\Pi_\phi|x}^2\Big] = \frac{(2k)!}{d^{(2k)}} \braket{x\otimes x|\Pi_{\rm sym}^{(2k)}|x \otimes x}  \leq 4^k\left(\frac{k!}{d^{(k)}}\right)^2,
	\end{align}
	This yields
	\begin{align}
		\|\braket{x|\Pi_\phi|x}\|_{L^s} &\leq 2^{k(1 - 1/s)}\frac{k!}{d^{(k)}}  & \forall s \in [1,2] 
	\end{align}
	The other factors in the summand of Eq.~\eqref{eq:ScrDiff} can be bounded using Lemma \ref{lem:RConc}.
	\begin{align}
		\|r_{\phi}^{1-k+a}\|_{L^q} &\leq \left[\exp\left(\frac{q^2(k-a-1)^2}{2(m-q(k-a-1))} \right)\right]^{1/q} \leq e^{q k^2/m} & \text{assuming }qk &\leq m/2\\
		\|r_\phi-1\|_{L^p} &\leq \left[4 \frac{\Gamma(p+1)}{(m/4)^{p/2}}\right]^{1/p} \leq \frac{2p}{\sqrt{m}} & p &\geq 4
	\end{align}
	where we have use the bound $4\Gamma(p+1) \leq p^p$ which holds for $p \geq 4$. Since we can assume $k \geq 2$, we choose $s = k/(k-1)$, $q = p = 2k$, then noticing that  $\braket{x|\omega^{(k)}_{\mathcal{S}_\rho, {\rm Un}}|x} = \braket{x|\chi^{(k)}_{\rm Haar}|x} = \frac{k!}{d^{(k)}}$, we can bring all these inequalities together to obtain
	\begin{align}
		|\braket{x|\omega^{(k)}_{\mathcal{S}_\rho, {\rm Un}}-\omega^{(k)}_{\mathcal{S}_\rho}|x}| &\leq \braket{x|\omega^{(k)}_{\mathcal{S}_\rho, {\rm Un}}|x} \times \frac{11k^2}{\sqrt{m}} & \text{for }k \leq m^{1/3}/2.
	\end{align}
	Thus, for the stated range of $k$, we have proved Eq.~\eqref{eq:RelErrorConjugated} for $\delta = \frac{11k^2}{\sqrt{m}}$. Recalling that $m$ can be any integer less than $2^{S_{\infty}(\rho)} = (\max \text{eig }\rho)^{-1}$, and also using the approximation \eqref{eq:chi approx def 2}, which has relative error $\mathcal{O}(k^2/d) \leq \mathcal{O}(k^2 2^{-S_\infty(\rho)})$, this implies the relative error bound stated in Theorem \ref{thm:approx}. \qed

	\section{Concentration of all local properties} \label{app:concentration}
	
	In this section, we consider the moments of the Scrooge ensemble after tracing out small subregions.
	We first present and prove our main result, which shows that the moments of the Scrooge ensemble rapidly collapse to the moments, $\rho^{\otimes k}$, of the background state $\rho$ upon tracing out small subregions.
	This holds up to an approximation error that, for typical $\rho$, decays exponentially in the size of the subregion that is traced out.
	We then apply this result to show that Scrooge-random states cannot be distinguished from the background state $\rho$ after a small amount of local depolarizing noise has been applied. More generally, since depolarizing noise can be thought of as equivalent to tracing out random small subregions, we expect to see a similar correspondence between the behavior of Scrooge-random states under noise vs.~after tracing out qubits.
	
	\subsection{Approximation of subsystem moments}
	
	Here we substantiate the claims made in the main text with rigorous bounds on the relative error of the approximation in Eq.~\eqref{eq:subsystem approx}. This will be used extenstively in subsequent sections. 
	\begin{theorem} \label{thm:subsystem full}
		For $k \geq 0$, let $\{A_r\}_{r=1}^k$ be any collection of subregions, with complements $\{B_r\}_{r=1}^k$. The $k$-th subsystem moments of the Scrooge ensemble $\chi^{(k)}_{A_1, \cdots, A_k}$, obtained by tracing out regions $B_r$ from the full moments $\chi_{\mathcal{S}_\rho}^{(k)}$, can be approximated by
		\begin{align}
			\chi^{(k)}_{A_1, \cdots, A_k} \coloneqq \tr_{B_1 \cdots B_k}[\chi_{\mathcal{S}_\rho}^{(k)}] \approx  \tr_{B_1 \cdots B_k}\big[\rho^{\otimes k}\big] 
			\label{eq:reduced approx full}
		\end{align}
		up to relative error
		\begin{align}
			\delta[A_1, \ldots, A_k] = O(k^2 2^{-\min_{r} \hat{S}_\infty(B_r|A_r)}).
			\label{eq:subsystem error general}
		\end{align}
		where $\hat{S}_\infty(B_r|A_r)_\rho \geq 0$ [defined below in Eq.~\eqref{eq:entropy hat def}] is an entropy-like quantity evaluated on a single copy of the state $\rho \equiv \rho^{A_rB_r}$. For moments that have a product structure across different copies, the relative error
		\begin{align}
			\delta_{\rm prod}[A_1, \ldots, A_k] \coloneqq \max_{\ket{x} = \ket{x_1 \otimes \cdots \otimes x_r}}\frac{\left|\braket{x|\chi_{A_1, \ldots A_k}^{(k)}|x} - \braket{x|\tr_{B_1 \cdots B_k}\big[\rho^{\otimes k}\big]|x}\right|}{\braket{x|\tr_{B_1 \cdots B_k}\big[\rho^{\otimes k}\big]|x}}
			\label{eq:epsilonProd}
		\end{align}
		improves to $\epsilon_{\rm prod} = O(k^2 2^{-\min_{r} S_\infty^*(B_r)})$, where $S_\infty^*(B_r) $ is the minimum possible post-measurement min-entropy on $B_r$ after performing a measurement on $A_r$, starting from the state $\rho^{A_rB_r}$ [Eq.~\eqref{eq:post meas def}]. 
	\end{theorem}  \noindent The case where all $\{A_r\}_{r=1}^k$ are equal to a single subregion $A$ corresponds to Eq.~\eqref{eq:subsystem approx} in the main text.\\

	We argue below that for typical background states $\rho$, the entropic quantities $\hat{S}_\infty(B_r|A_r)$ scale linearly in the number of qubits in $B_r$. Thus, we expect that an approximation accuracy $\epsilon$ will be achieved when the number of qubits in the regions $B_r$ scales as $\Omega(\log(1/\epsilon))$, which is independent of the system size for a constant target accuracy $\epsilon$. Surprisingly, after tracing out a vanishing fraction of the total number of qubits, the low-order moments become close to those that would be obtained by replacing the random states $\dyad{\psi}$ with the fixed bakcground density matrix $\rho$. Thus, instance-to-instance fluctuations of Scrooge-random states can only be seen in truly global quantities, such as measurement outcome distributions (see Section \ref{app:output distributions}).
	
	Our characterisation of finite moments of reduced states should be distinguished from standard formulations of subsystem thermalization in closed quantum systems, where $A$ is typically smaller than the traced out region $B$. In that spirit, a recent work has shown that Scrooge-random states exhibit \textit{canonical typicality}: the reduced state $\tr_B \psi $ is close to $\tr_B(\rho)$ in trace distance with high probability, provided $B$ scales linearly with system size \cite{teufel2025canonical}. In contrast, Eq.~\eqref{eq:reduced approx full} continues to hold for much smaller regions $B$, as small as $|B| =\Theta( \log 1/\epsilon)$. Indeed, relative-error closeness of the moments of two ensembles $\{(p_i, \rho_i)\}$, $\{(p_i, \sigma_i)\}$ is an incomparable condition to closeness of the constituent states in trace distance, $\|\rho_i - \sigma_i\|_1$. Moments capture the statistics of expectation values $\tr[O \rho_i]$ for observables $O$ that are not themselves dependent on $\rho_i$. In contrast, the trace distance between the states in two ensemble $\| \rho_i - \sigma_i\|_1 = \sup_{O_i: \|O_i\|_\infty = 1}\tr[O_i(\rho_i - \sigma_i)]$ captures the distinguishability of the ensembles using statistical tests that leverage information about specific instances of states in the ensembles $\rho_i$, $\sigma_i$. This explains why the approximation \eqref{eq:reduced approx full} can be accurate in regimes where we do not expect canonical typicality to hold. See also Footnote~\cite{subsystemFootnote} in the main text.

	\subsection{Properties of entropic quantities $S^*_\infty(B)$ and $\hat{S}_\infty(B|A)$}\label{app:entropy}
	
	In Theorem \ref{thm:subsystem full}, the relative error of the approximation \eqref{eq:reduced approx full} for the $k$th moments is controlled by a single-copy entropic quantity $\hat{S}_\infty(B_r|A_r)_\rho$ evaluated on the background state $\rho$. We provide the full definition of this quantity and some of its properties here. For a given copy of the system $r$, we have
	\begin{align}
		\hat{S}_\infty(B_r|A_r)_\rho \coloneqq -\log_2 \left\| (\rho^{A_r})^{-1/2} (\rho^{A_rB_r})^{T_{B_r}} (\rho^{A_r})^{-1/2} \right\|_\infty,
		\label{eq:entropy hat def}
	\end{align}
	where $T_{B_r}$ denotes a partial transpose on $B_r$.
	For product moments, the relative error $\delta_{\rm prod}$ [Eq.~\eqref{eq:epsilonProd}] can be improved by replacing the above quantity with the \textit{minimum post-measurement min-entropy}
	\begin{align}
		S_\infty^*(B_r)_\rho &\coloneqq \min_{y_r \in \text{sup}[\rho^{A_r}]} \left(-\log_2 \left\| \rho^{B_r}_{y_r} \right\|_\infty\right)
		\label{eq:post meas def}
	\end{align}
	Here, $\rho^{B_r}_{y_r}$ is the conditional state on $B_r$ after performing a measurement of $A_r$ on the background state $\rho$, and obtaining an outcome that corresponds to a POVM element $\dyad{y_r}$ [see Eq.~\eqref{eq:cond state def}]. 
	For convenience, we will drop the label $r$ for the rest of this section; thus $A$ and $B$ will refer to single-copy Hilbert spaces.
	
	Here we list some basic facts about $\hat{S}_\infty(B|A)$ and $S^*_\infty(B)$. Overall, these suggest an interpretation of each as a measure of the entropy in subsystem $B$ that cannot be attributed to correlations with qubits in $A$.
	
	\newcommand{\itemeqnumber}[1][]{\hfill\refstepcounter{equation}\textup{(\theequation)}{\label{#1}}}
	\begin{itemize}
		\item $0 \leq \hat{S}_\infty(B|A)_{\rho} \leq S^*_\infty(B)_{\rho} \leq S_\infty(\rho^B)$ \itemeqnumber[eq:entropy nonneg]
		\item $\hat{S}_\infty(B|A) = S^*_\infty(B) = S_\infty(\rho^B)$ if $\rho = \rho^A \otimes \rho^B$ \itemeqnumber[eq:entropy product]
		\item $\hat{S}_\infty(B_1B_2|A_1A_2)_{\rho^{A_1B_1} \otimes \sigma^{A_2B_2}} = \hat{S}_\infty(B_1|A_1)_{\rho^{A_1B_1}} + \hat{S}_\infty(B_2|A_2)_{\sigma^{A_2B_2}}$. \itemeqnumber[eq:entropy additive]
		\item $\hat{S}_\infty(B|A)_\rho$ is quasi-concave in $\rho$, i.e.~for $\rho = \sum_i p_i \rho_i$ with $p_i \geq 0$, $\sum_i p_i = 1$, and $\rho_i$ a set of density operators, we have $\hat{S}_\infty(B|A)_\rho \geq \min_i  \hat{S}_\infty(B|A)_{\rho_i}$ \itemeqnumber[eq:entropy convex]
		\item For classical states, where $\rho^{AB} = \sum_{ij} p_{ij} \dyad{i}_A \otimes \dyad{j}_B$ for some pair of orthogonal bases $\{\ket{i}_A\}_i$, $\{\ket{j}_B\}_i$ and probabilities $p_{ij}$, we have $\hat{S}_\infty(B|A) = S^*_\infty(B) = H^\downarrow_\infty(B|A)_{p_{ij}}$, where $H^\downarrow_\infty(B|A)_{p_{ij}} = -\log(\max_{ij} p_{j|i})$ is the classical conditional min-entropy, with $p_{j|i} \coloneqq p_{ij}/(\sum_{j'}p_{ij'})$ the conditional probabilities. \itemeqnumber[eq:entropy classical]
	\end{itemize}
	We also note that the definition of $\hat{S}_\infty(B|A)$ bears a close resemblance to the quantum conditional min-entropy $H^{\downarrow}_\infty(B|A)_\rho \coloneqq -\log_2 \|(\rho^{A})^{-1/2} \rho^{AB} (\rho^A)^{-1/2}\|_\infty$ \cite{tomamichel2016}, which differs only by the presence of a partial transpose in Eq.~\eqref{eq:entropy hat def}; indeed, Property \eqref{eq:entropy classical} can be shown straightforwardly from this observation. Unlike $\hat{S}_\infty(B|A)$, the conditional min-entropy can be negative (but only when $\rho^{AB}$ has bipartite distillable entanglement).
	
	Properties \eqref{eq:entropy product} and \eqref{eq:entropy additive} follow immediately from the definition of $\hat{S}_\infty(B|A)_\rho$. In particular, these imply that if we extend $B$ by tensoring in a mixed state $\sigma^{B_2}$, then $\hat{S}_\infty(B|A)$ will increase by an amount $S_\infty(\sigma^{B_2})$. Together with the non-negativity of $\hat{S}_\infty(B|A)$, this suggests that $\hat{S}_\infty(B|A)_\rho$ will be extensive in the number of qubits of $B$ in typical many-body mixed states, which typically have finite entropy density and weak correlations between distant degrees of freedom. 
	
	To prove Property \eqref{eq:entropy nonneg}, we provide the following alternative characterisation of $\hat{S}_\infty(B|A)_\rho$. For any Kraus operator $K^{A \rightarrow B'}$, which maps states on $A$ to states on a Hilbert space $B'$ of the same dimension as $B$, we define
	\begin{align}
		\rho_K^{BB'} = \frac{K^{A \rightarrow B'}\rho^{AB}(K^{A \rightarrow B'})^\dagger}{\tr[K^{A \rightarrow B'}\rho^{A}(K^{A \rightarrow B'})^\dagger]}
	\end{align}
	which is evidently a properly normalised state. Then we have
	\begin{align}
		\hat{S}_\infty(B|A)_\rho = -\log_2 \max_{K^{A \rightarrow B'}} \left|\tr[F^{BB'}\rho^{BB'}_K]\right|
		\label{eq:entropy hat swap}
	\end{align}
	where $F^{BB'} = \sum_{i,j=1}^{d_B} \ket{i}\bra{j}_{B}\otimes \ket{j}\bra{i}_{B'}$ is the swap operator between Hilbert spaces $B$ and $B'$. Eq.~\eqref{eq:entropy hat swap} can be shown using the fact that any state $\ket{x}_{AB}$ can be written as $(K^{A \rightarrow B'})^\dagger \ket{\Omega}_{BB'}$, where $\ket{\Omega}_{BB'} = \sum_{j=1}^{d_B}\ket{j}_B \otimes \ket{j}_{B'}$, along with the identity $(\ket{\Omega}\bra{\Omega}_{BB'})^{T_B} = F_{BB'}$. If we were to restrict the maximization over $K^{A \rightarrow B'}$ to rank-1 operators $K^{A \rightarrow B'} = \ket{x_{B'}}\bra{y_A}$, we would obtain $S_\infty^*(B)$, and so $\hat{S}_\infty(B|A) \leq S^*_\infty(B)$. Moreover, since $F^{BB'}$ has unit spectral norm, H{\"o}lder's inequality applied to the above trace gives us non-negativity of $\hat{S}_\infty(B|A)$. Finally, for any $\ket{y}$ there exists a rank-1 POVM $\{\ket{y_i}\bra{y_i}\}_i$ which contains $\ket{y_0}\bra{y_0} = \ket{y}\bra{y}$ as the $i = 0$-th element. Then, $\rho^{B} = \sum_i p_i \rho^B_{y_i}$, where $p_i$ are the outcome probabilities for this POVM, and by quasi-convexity of the spectral norm we have $\|\rho^B\|_\infty \leq \|\rho^{B}_{y}\|_\infty$. This proves property \eqref{eq:entropy nonneg}.
	
	Finally, property \eqref{eq:entropy convex} follows from a second alternative characterization of the norm in \eqref{eq:entropy hat def}.
	\begin{align}
		\|(\rho^{A})^{-1/2}(\rho^{AB})^{T_B}(\rho^{A})^{-1/2}\|_\infty = \inf\Big\{ \lambda \,\Big|\, -\lambda \rho^A \preceq (\rho^{AB})^{T_B} \preceq \lambda \rho^A   \Big\}
		\label{eq:entropy had inf}
	\end{align}
	If $\rho^{AB} = \sum_i p_i \sigma_i^{AB}$ for a set of probabilities $p_i \geq 0$, then similarly $\rho^A =\sum_i p_i \sigma_i^A$ where $\sigma^A_i = \tr_B[\sigma^{AB}_i]$. We then $\min_i \hat{S}_\infty(B|A)_{\sigma_i} = \inf\{\lambda\, |\, -\lambda \sigma_i^A \preceq \sigma^{AB}_i \preceq \lambda \sigma_i^A\, \forall i\}$. The condition in this set implies $\sum_i p_i (-\lambda \sigma^{A}_i) \preceq \sum_i p_i \sigma^{AB}_i \preceq \sum_i p_i \lambda \sigma^A_i$, and this proves \eqref{eq:entropy convex}.


	\subsection{Collapse of the Scrooge ensemble under small amounts of noise}
	
	We can now establish the following collapse of Scrooge-random output distributions under small amounts of noise.
	\begin{corollary}
		[Scrooge-random output distributions collapse under noise] \label{cor: noise}
		Let $\psi \sim \mathcal{E}$ be drawn from an ensemble $\mathcal{E}$ that forms an $\epsilon$-approximate Scrooge $k$-design with relative error $\epsilon$, and let $\tilde{\psi} = \mathcal{D}_\gamma(\chi)$ denote the  same state acted on by depolarizing noise of strength $\gamma$ on each qubit.
		Let $\tilde{\rho} = \mathcal{D}_\gamma(\rho)$ denote the density matrix $\rho$ acted on by the same noise channel.
		Finally, let $S^*_k(\ell) = \min_{B : |B| \leq \ell} S^*_\infty(B)$ be the smallest post-measurement entropy over all regions $B$ of size at most $\ell$ [with $S^*_\infty(B)$ defined in \eqref{eq:post meas def}].
		
		Consider any quantum experiment that prepares either $k$ copies of an unknown instance of $\tilde{\psi}$, or $k$ copies of $\tilde{\rho}$, and performs any sequence of measurements on each copy.
		No statistical test can distinguish the measurement outcomes from the states $\tilde{\psi}$ and $\tilde{\rho}$ with advantage greater than $k e^{-\mathcal{O}(\gamma n)} +  k^2 e^{-S^*_k(\gamma n/4)} + \epsilon$.
	\end{corollary}
	\noindent In a typical many-body quantum mixed state, the entropy will grow linearly in the size of the subregion, $S_k^*(\gamma n/4) = \Theta(\gamma n)$.
	In this setting, Corollary~\ref{cor: noise} proves that any experiment to distinguish $\tilde{\chi}$ and $\tilde{\rho}$ requires a number of samples growing exponentially in $\gamma n$.
	This becomes super-polynomial in the number of qubits $n$ already at macroscopically small noise rates, $\gamma = \omega(\log n / n)$.
	This mirrors analogous results (which are much simpler to derive) for noisy Haar-random states.
	
	Corollary \ref{cor: noise} describes a scenario where we do not possess any prior knowledge about the particular instance $\psi$, in that the choice of measurement bases and any subsequent post-processing of the outcomes can only depend on $\rho$. This model is likely to be appropriate in settings where an implicit description of the state is available, but computing its properties is classically intractable---for instance in the temporal ensemble, where one knows the time of evolution $t$, but where simulating the time-evolved state $\ket{\psi_t} = e^{-\iu H t}\ket{\psi_0}$ may be hard. In Section \ref{app:output distributions}, we describe different notions of distinguishability in which one obtains noisy samples from a \textit{known} instance $\psi$, thereby allowing properties of $\psi$ to be invoked during the classical post-processing (as one does when computing cross-entropy benchmarks \cite{boixo2018characterizing}). There we find that noisy Scrooge-random states still remain indistinguishable from the background state $\rho$ (see Corollary \ref{cor:TVDclose}).
	
	\subsection{Proof of Theorem~\ref{thm:subsystem full}: Moments of the Scrooge ensemble after tracing out small subregions}
	
	\textit{Proof of Theorem \ref{thm:subsystem full}.---}The relative error $\epsilon$ of the approximation \eqref{eq:reduced approx full} can be characterised in the following way
	\begin{align}
		\epsilon \coloneqq \max_{\substack{\ket{a} \in \text{sup}[\rho^A]}} \left| \frac{\mathbb{E}_{\psi} \braket{a| \tr_{B}[\psi^{\otimes k}]|a } }{\braket{a|\rho^A |a} } - 1 \right|.
	\end{align}
	For the purposes of this proof only, we are using the shorthand $A$ to denote the union of all regions $\{A_r\}$, and similarly for $B$; accordingly we define $\rho^A \coloneqq \bigotimes_{r=1}^k \tr_{B_r}[\rho]$. 
	Using the approximate formula for the $k$th Scrooge moments [Eq.~\eqref{eq:approx}], and writing $ \rho^{\otimes k} \equiv \rho^{AB}$, we have
	\begin{align}
		\epsilon'\coloneqq\frac{1+\epsilon}{1+\delta_\rho}-1\leq \max_{\substack{\ket{a} \in \text{sup}[\rho^A]}} \left| \frac{\sum_{\pi \in S_{k}} \tr[\big(\ket{a}\bra{a}_A \otimes I^B\big) \pi \rho^{AB}]}{\braket{a|\rho^A |a}} - 1 \right|.
		\label{eq:rel error subsystem approx}
	\end{align}
	We can substitute the argument of the maximum for $\ket{x} = (\rho^A)^{1/2}\ket{a} \in \text{sup}[\rho^A]$, such that the denominator can be written as $\braket{a|\rho^A|a} = \|\ket{x}\|^2$. Then, noting that the term $\pi = I$ exactly cancels with the $-1$ term, we can apply the triangle inequality to obtain
	\begin{align}
		\epsilon' &\leq \sum_{\substack{\pi \in S_k\\\pi \neq I}} \max_{\substack{\ket{x} \in \text{sup}[\rho^A] \\ \|\ket{x}\| = 1}} \tr[\big((\rho^{A})^{-1/2}\ket{x}\bra{x}_A(\rho^{A})^{-1/2} \otimes I^B\big) \pi \rho^{AB}] \nonumber\\
		&= \sum_{\substack{\pi \in S_k\\\pi \neq I}} \left\|\tr_B[ (\rho^A)^{-1/2} (\rho^{AB})^{1/2} \pi (\rho^{AB})^{1/2}(\rho^A)^{-1/2}]  \right\|_\infty,
		\label{eq:subsystem epsilon q channel}
	\end{align}
	where we use the fact that $\pi$ commutes with $\rho^{AB}$ in the last equality. Now, the spectral norm $\|\,\cdot\,\|_\infty$ is multiplicative under tensor products, and the argument of the norm can be written as a tensor product over the cycles of $\pi$, since we can always write $\pi = \bigotimes_{\ell \in \pi} \pi_\ell$, where $\ell$ labels different cyclic permutation operators $\pi_\ell$, each acting on $l \equiv |\ell|$ different copies. Therefore, we have
	\begin{align}
		\epsilon' \leq \sum_{\pi \neq I} \prod_{\ell \in \pi} \max_x \epsilon_{\ell}[x],
		\label{eq:epsilon cycle product}
	\end{align}
	where $\epsilon_{\ell}[x] \coloneqq \tr[(\ket{x}\bra{x}_A\otimes I^B)\cdot (\rho^A)^{-1/2} (\rho^{AB})^{1/2} \pi_\ell (\rho^{AB})^{1/2}(\rho^A)^{-1/2}]$. Here and subsequently, regions $A, B$ are implicitly constructed from $l \leq k$ copies of the system.
	
	To deal with the cyclic permutation operator $\pi_\ell$, we will use the identity
	\begin{align}
		\tr_{Q_1 \cdots Q_l}\left[(\pi_\ell^{Q_1 \cdots Q_l}\otimes I^{Y_1 \cdots Y_l}) \bigotimes_{r=1}^l X_r^{Q_rY_r}\right] = \tr_{Q_1}\left[\prod_{r=1}^lX_r^{Q_1Y_r} \right],
		\label{eq:perm identity}
	\end{align}
	for any set of operators $X_r^{Q_rY_r}$ acting on a collection of independent Hilbert spaces $\{Q_r\}_{r=1}^l$, $\{Y_r\}_{r=1}^l$, with all $Q_r$ having the same dimension. On the right hand side, the notation $X^{Q_1Y_r}_r$ describes a set of transformed operators that act on different spaces $Y_r$ but the same space $Q_1$. The order of the product should also follow the order of the cyclic permutation.
	
	While $\rho^{AB}$ is a tensor product across different copies of the system, $\ket{x}$ is not necessarily so. To help us invoke Eq.~\eqref{eq:perm identity}, we therefore use the trivial identity $\ket{x}_{A_1 \cdots A_l} = (\bra{\bar{x}}_{A'}\otimes I^A)\left(\bigotimes_{r=1}^l \ket{\Omega}_{A_rA_r'}\right)$, where $\ket{\Omega}_{A_rA_r'} = \sum_{j_r=1}^{d_{A_r}}\ket{j_r}_{A_r}\otimes \ket{j_r}_{A'_r}$ is an unnormalized EPR state, and $\ket{\bar{x}}_{A'} \equiv (\bra{x}_{A'})^{T_{A'}}$ is the complex conjugate of $\ket{x}_{A'}$ in the basis $\{\otimes_{r=1}^l\ket{j_r}_{A'_r}\}$ ($T_{A'}$ denotes a partial transpose on $A'$ in the same basis). Defining operators $K^{B_r \rightarrow A_rB_rA'_r} \coloneqq  (\rho^{AB})^{1/2} (\rho^A)^{-1/2}(\ket{\Omega}_{A_rA'_r}\otimes I^{B_r})$ and $K^{B \rightarrow ABA'} = \bigotimes_{r=1}^l K^{B_r \rightarrow A_rB_rA'_r}$, we can then write $\epsilon_\ell[x]$ in a form that allows us to use  \eqref{eq:perm identity} with $Y_r = A_r'$, $Q_r = A_r \cup B_r$.
	\begin{align}
		\epsilon_\ell[x] &= \Braket{\bar{x}| \tr_{AB}\Big[K^{B \rightarrow ABA'}(K^{B \rightarrow ABA'})^\dagger (\pi^{AB}_\ell \otimes I^{A'})\Big] |\bar{x}}_{A'} \nonumber\\
		&=\Braket{\bar{x}|\tr_{A_1B_1}\Bigg[\prod_{r=1}^l \big(K^{B_1 \rightarrow A_1B_1A'_r}(K^{B_1 \rightarrow A_1B_1A'_r})^\dagger  \otimes I^{A' \backslash A_r'}\big)\Bigg]|\bar{x}}_{A'} \label{eq:subsystem prod k}.
	\end{align}
	where each operator $K^{B_1 \rightarrow A_1B_1A'_r}$ acts on the same copy $A_1B_1$, but different copies of $A'_r$. We now rewrite the inner product using $\braket{\bar{x}|X^{A'}|\bar{x}} = \tr[\ket{\bar{x}}\bra{\bar{x}}_{A'} X^{A'}]$, and take the trace over $A_1'$ separately to get
	\begin{align}
		\epsilon_\ell[x] &= \tr_{A_1B_1 A'_2 \cdots A'_l}\left[\sigma^{A_1B_1A_2'\cdots A_l'}_x \prod_{r=2}^l\big(K^{B_1 \rightarrow A_1B_1A'_r}(K^{B_1 \rightarrow A_1B_1A'_r})^\dagger \otimes I^{A' \backslash A_1'A_r'}\big) \right] 	\label{eq:epsilon k prod}  \\  \text{where } 	\sigma^{A_1B_1A_2'\cdots A_l'}_x &\coloneqq \tr_{A'_1}\left[\big(\ket{\bar{x}}\bra{\bar{x}}_{A'_1 \cdots A'_l}\otimes I^{A_1B_1}\big) K^{B_1 \rightarrow A_1B_1A'_1}(K^{B_1 \rightarrow A_1B_1A'_1})^\dagger \right]
		\nonumber\\ &=
		(V_x^{A_1 \rightarrow A'_2 \cdots A'_l})^\dagger K^{B_1 \rightarrow A_1B_1A'_1}(K^{B_1 \rightarrow A_1B_1A'_1})^\dagger V_x^{A_1 \rightarrow A'_2 \cdots A'_l}
		\label{eq:sigma pos}
	\end{align}
	In the last line, we define $V_x^{A_1 \rightarrow A_2'\cdots A_l'} = (\bra{\Omega}_{A_1A_1'} \otimes I^{A_2' \cdots A_l'})\cdot (\ket{\bar{x}}_{A_1'\cdots A_l'} \otimes I^{A_1})$, and use the identity $\tr_{A'_1}[X^{A_1' \cdots A_l'}] = (\bra{\Omega}_{A_1A_1'}\otimes I^{A_2'\cdots A_l'})X^{A_1' \cdots A_l'} (\ket{\Omega}_{A_1A_1'}\otimes I^{A_2'\cdots A_l'})$ for an arbitrary operator $X^{A_1' \cdots A_l'}$.
	
	The final line in Eq.~\eqref{eq:sigma pos} shows that $\sigma^{A_1B_1A_2'\cdots A_l'}_x$ is positive semi-definite.
	Moreover, since $\tau^{A_1}_x \coloneqq V_x^{A_1 \rightarrow A_2'\cdots A_l'} (V_x^{A_1 \rightarrow A_2'\cdots A_l'})^\dagger$ has unit trace (as one can easily verify), we have
	\begin{align}
		\tr[\sigma^{A_1B_1A_2'\cdots A_l'}_x] &= \tr_{A_1B_1A'_1}[\tau^{A_1}_xK^{B_1 \rightarrow A_1B_1A'_1}(K^{B_1 \rightarrow A_1B_1A'_1})^\dagger ]
		\nonumber\\ &=\tr_{A_1B_1A'_1}[\tau^{A_1}_x (\rho^{A_1B_1})^{1/2} (\rho^{A_1})^{-1/2}\big(\dyad{\Omega}_{A_1A_1'}\otimes I^{B_1}\big) (\rho^{A_1})^{-1/2}  (\rho^{A_1B_1})^{1/2} ] \nonumber\\
		&= \tr_{A_1B_1}[\tau^{A_1}_x (\rho^{A_1B_1})^{1/2} (\rho^{A_1})^{-1}  (\rho^{A_1B_1})^{1/2} ] \nonumber\\
		&= \tr_{A_1}[\tau_x^{A_1}] = 1.
	\end{align}
	In the third equality, we have used the fact that $\tr_{A_1'}[\dyad{\Omega}_{A_1A_1'}] = I^{A_1}$. Hence, $\sigma^{Q_1A_2'\cdots A_l'}_x$ is a properly normalised density operator. Applying H{\"o}lder's inequality in the form $\|X \prod_{r} Y_r\|_1 \leq \|X\|_1 \prod_r \|Y_r\|_\infty$ to \eqref{eq:epsilon k prod}, we then find
	\begin{align}
		\epsilon_\ell[x] &\leq \prod_{r=2}^l \|K^{B_r \rightarrow A_rB_rA'_r}(K^{B_r \rightarrow A_rB_rA'_r})^\dagger\|_\infty \nonumber\\
		&= \prod_{r=2}^l \left\|\big((\rho^{A_r})^{-1/2} \rho^{A_rB_r} (\rho^{A_r})^{-1/2}\big)^{T_{A_r}}\right\|_\infty \eqqcolon \prod_{r=2}^l \mu[A_r]
		\label{eq:epsilon PT}
	\end{align}
	The second line is a consequence of the following property of $K^{B_r \rightarrow A_rB_rA'_r}$
	\begin{align}
		(K^{B_r \rightarrow A_rB_rA'_r})^\dagger K^{B_r \rightarrow A_rB_rA'_r} = (\bra{\Omega}_{A_rA_r'}\otimes I^{B_r})(\rho^{A_r})^{-1/2}\rho^{A_rB_r}(\rho^{A_r})^{-1/2}(\ket{\Omega}_{A_rA_r'}\otimes I^{B_r}),
	\end{align}
	which is unitarily equivalent to the partially transposed operator appearing in the argument of the norm in Eq.~\eqref{eq:epsilon PT}. 
	
	Letting $\mu^* \coloneqq \max_{r} \mu[\mathcal{N}_r] = 2^{-\min_r \hat{S}_\infty(B_r|A_r)_\rho}$, where the function $\hat{S}_\infty$ is defined in Theorem \ref{thm:subsystem full}, we can use \eqref{eq:epsilon cycle product} to upper bound the error as
	\begin{align}
		\epsilon' \leq \sum_{\substack{\pi \in S_k \\ \pi \neq I}} \prod_{\ell \in \pi} (\mu^*)^{|\ell| - 1} = \sum_{\substack{\pi \in S_k \\ \pi \neq I}} (\mu^*)^{-|\pi|} \leq \exp(\frac{k(k-1)}{2} \mu^* ) - 1 \leq k(k-1)\mu^* = O(k^2 \mu^*).
		\label{eq:subsystem epsilon sum perm}
	\end{align}
	where $|\pi| = \sum_{\ell \in \pi} 1$ is the number of cycles in $\pi$, and the latter inequality holds provided $k(k-1)\mu^* \leq 1$. In the last step, we use inequality (13) from Ref.~\cite{harrow2021approximate}. As discussed in Section \ref{app:entropy}, we define $\hat{S}_\infty(B_r|A_r) \coloneqq -\log_2 \|(\rho^{A_r})^{-1/2} (\rho^{A_rB_r})^{T_{B_r}} (\rho^{A_r})^{-1/2}\|_\infty$, and since $\|(X^{A_rB_r})^{T_{A_r}}\|_\infty = \|(X^{A_rB_r})^{T_{B_r}}\|_\infty$ for all operators $X^{A_rB_r}$, this proves the first statement of Theorem \ref{thm:subsystem full}.

	In the special case where $\ket{x} = \bigotimes_{r=1}^k \ket{x_r}_{A_r}$ is a product state across different copies,  Eq.~\eqref{eq:subsystem prod k} simplifies to
	\begin{align}
		\epsilon_\ell[x] &= \tr\left[\prod_{r=1}^k L_{x_r}^{A_1B_1}(L_{x_r}^{A_1B_1})^\dagger  \right] &
		\text{where }L_{x_r}^{A_1B_1} &\coloneqq \big(\bra{\bar{x}_r}_{A_r'} \otimes I^{A_1B_1}\big) K^{B_1 \rightarrow A_1B_1 A_r'}
		\label{eq:subsystem product Lx}
	\end{align}
	We note that $(L_{x_r}^{A_rB_r})^\dagger L_{x_r}^{A_rB_r}$, which has the same eigenvalues as  $L_{x_r}^{A_1B_1} (L_{x_r}^{A_1B_1})^\dagger$, can be written as
	\begin{align}
		(L_{x_r}^{A_rB_r})^\dagger L_{x_r}^{A_rB_r} &= \rho^{B_r}_{y_r}  & \text{where }\ket{y_r} \coloneqq \frac{(\rho^{A_r})^{-1/2}\ket{x_r}}{\|(\rho^{A_r})^{-1/2}\ket{x_r}\|}. \label{eq:subsystem epsilon productX}
	\end{align}
	Here $\rho^{B_r}_{y_r}$, is the post-measurement state on $B_r$ after measuring on $A_r$ with an outcome with POVM element $\ket{y_r}\bra{y_r}_{A_r}$, namely
	\begin{align}
		\rho^{B_r}_{y_r} &\coloneqq \frac{\tr_{A_r}[(\ket{y_r}\bra{y_r}_{A_r}\otimes I^{B_r})\rho^{A_rB_r}]}{\braket{y_r|\rho^{A_r}|y_r}} = \tr[(\ket{x_r}\bra{x_r}\otimes I^{B_r})(R^{Q_r|A_r})^\dagger R^{Q_r|A_r}],
		\label{eq:cond state def}
	\end{align}
	which are evidently normalized density operators. Applying H{\"o}lder's inequality to Eq.~\eqref{eq:subsystem product Lx} in a similar fashion to Eq.~\eqref{eq:epsilon PT}, we find
	\begin{align}
		\epsilon_\ell[x] &\leq \prod_{r=2}^l \|\rho^{B_r}_{y_r}\|_\infty \leq 2^{-(l-1)\min_r S^*_\infty(B_r)}
	\end{align}
	where we define $S^*_\infty(B_r) \coloneqq \max_{y_r} S_\infty(\rho^{B_r}_{y_r})$, as elaborated on below. Using the same inequalities as in Eq.~\eqref{eq:subsystem epsilon sum perm} with $2^{-\min_r S^*_\infty(B_r)}$ in place of $\mu^*$, we establish the second part of Theorem \ref{thm:subsystem full}. \hfill $\square$

	\subsection{Proof of Corollary~\ref{cor: noise}: Collapse of the Scrooge ensemble under small amounts of local depolarizing noise}
	
	Based on the stated assumptions, our aim is to determine whether one can use adaptive measurements to distinguish between $k$ copies of an unknown Scrooge-random state $\psi \sim \mathcal{S}_\rho$ from $k$ copies of $\rho$, after local depolarizing noise has been applied. Formally, the $k$-copy state to which the  distinguisher has access is either $\tilde{\chi} = \E_{\psi \sim \mathcal{S}_\rho} \mathcal{D}_\gamma[\psi]^{\otimes k} = \mathcal{D}_\gamma^{\otimes k}[\chi^{(k)}_{\mathcal{S}_\rho}]$ in the former case, or $\tilde{\rho}^{\otimes k} \coloneqq \mathcal{D}_\gamma[\rho]^{\otimes k}$ in the latter. 
	
	Let $x_1, \ldots, x_k$ be a sequence of measurement outcomes obtained by the distinguisher, where $x_r$ labels an element $\ket{x_r}\bra{x_r}$ of a rank-1 projection-valued measure (PVM) on $n$ qubits in the $r$-th measurement round. The extension to more general postive operator-valued measures (POVMs) instead of PVMs follows via an identical computation with a moderate overhead in notation.
	We allow the PVM or POVM in the $r$-th measurement round to be chosen adaptively based on all measurement results received in previous rounds; that is, the basis of $\ket{x_r}$ may depend arbitrarily on $x_{r'}$ for all $r' < r$.
	Again, for brevity, we do not explicitly write this dependence in our notation. The probabilities of obtaining the particular sequence $x_1, \ldots, x_k$ in each setting are
	\begin{align}
		p^{\chi}(x_1, \ldots, x_k) &= \braket{x_1, \ldots, x_k|\tilde{\chi}|x_1, \ldots, x_k}  & p^\rho(x_1, \ldots, x_k) &= \braket{x_1, \ldots, x_k|\tilde{\rho}^{\otimes k}|x_1, \ldots, x_k}  \nonumber\\
		&= \E_{\psi \sim \mathcal{S}_\rho} \prod_{r=1}^k \braket{x_r|\mathcal{D}_\gamma[\psi]|x_r} & &= \prod_{r=1}^k \braket{x_r|\mathcal{D}_\gamma[\rho]|x_r}.
	\end{align}
	
	To show that the two possibilities cannot be distinguished by any statistical test, we must bound the total variation distance (TVD) between these $k$-round outcome distributions
	\begin{equation}
		\|\tilde{p}^\chi-\tilde{p}^\rho\|_{\rm TVD} \coloneqq \sum_{x_1,\ldots,x_k} \big| \tilde{p}^{\chi}(x_1,\ldots,x_k) - \tilde{p}^{\rho}(x_1,\ldots,x_k) \big|.
	\end{equation}.
	\noindent The TVD quantifies the advantage of any distinguisher with access to samples $x_1, \ldots x_k$ over an uninformed random guess: If both scenarios have the same \textit{a priori} probability, then the optimal success is $\frac{1}{2} + \frac{1}{4}\|\tilde{p}^\chi-\tilde{p}^\rho\|_{\rm TVD}$, where $\frac{1}{2}$ is the success probability for a random guess.
	
	Since $\mathcal{D}_\gamma$ is a tensor product of $n$ single-qubit depolarising channels, each of which can be written as $\mathcal{D}_{\gamma, i}[\sigma_i] = (1-\gamma)\sigma_i + \tr[\sigma_i] \mathbb{I}_i/2$, we can decompose the global noise channel into a mixture of complete depolarizing channels on subregions of qubits, namely $\mathcal{D}_\gamma[\sigma] = \sum_A q(A) \tr_A[\sigma]\otimes \mathbbm{1}^A/2^{|A|}$, where the sum is over all subregions $A \subseteq [n]$. Here, $q(A) = \gamma^{|A|} (1-\gamma)^{|\bar A|}$ can be interpreted as the probability that the noise channel $\mathcal{D}_\gamma$ acts solely on a subregion $A$.
	Accordingly,
	\begin{equation}
		\tilde{p}^\chi(x_1, \ldots, x_k) = \E_{\psi \sim \mathcal{S}_\rho} 
		\sum_{A_1,\ldots,A_k} 
		\prod_{i=1}^k  
		\left( \frac{q(A_i)}{2^{|A_i|}} \cdot \tr_{\bar A_i}\big( \tr_{A_i}(\dyad{x_i}) \cdot \tr_{A_i}( \dyad{\psi}) \big) \right).
	\end{equation}
	The summation over $A_1,\ldots,A_k$ can be decomposed as follows. 
	We set a threshold subregion size $\ell_* = \lceil\gamma n / 4\rceil$, and let $q(\ell_*)$ denote the probability that a given $A_i$ has size less than $\ell_*$ in the expansion of $\mathcal{D}_\gamma$.
	We use basic properties of the binomial distributions to show that $q(\ell_*) = e^{-\mathcal{O}(\gamma n)}$ in the final paragraph of the proof.
	By a union bound, the probability that all $k$ of $A_1,\ldots,A_k$ have size at least $\ell^*$ is therefore greater than $1 - k q(\ell_*)$. 
	
	From this decomposition, the total variational distance (TVD) between $\tilde{p}^\psi$ and $\tilde{p}^\rho$ can be bounded as
	\begin{equation}
		\|\tilde{p}^\chi-\tilde{p}^\rho\|_{\rm TVD} \leq 2 k q(\ell_*) + \sum_{A_1,\ldots,A_k} \prod_{i=1}^k \left( q(A_i) \cdot \delta_{|A_i| \geq \ell_*} \right) \cdot \sum_{x_1,\ldots,x_k} \left| \prod_{i=1}^k p^\rho(x_i;A_i)  - \E_{\psi \sim \mathcal{S}_\rho} \prod_{i=1}^k p^\psi(x_i;A_i) \right|,
	\end{equation}
	where $p^\psi(x_i;A_i) = \braket{x_i|\tr_{A_i}[\psi]\otimes \mathbbm{1}^{A_i}/2|x_i}$ is the probability of obtaining an outcome $x_i$ after applying the full depolarizing channel on region $A_i$ to the state $\dyad{\psi}$. The first term in the above upper bounds the TVD when at least one of $A_i$ is less than $\ell_*$, while the second term is equal to the TVD when all $A_i$ are greater than or equal to $\ell_*$. 
	From Theorem~\ref{thm:subsystem full}, the average of the absolute value in the second term over states $\psi \sim \mathcal{E}$ is $\eta$-small compared to $\prod_i p_\rho(x_i;A_i)$ in relative error, where $\eta \leq k(k-1) 2^{-\min_i S^*_\infty(A_i)}$. This yields
	\begin{equation}
		\|\tilde{p}^\chi-\tilde{p}^\rho\|_{\rm TVD} \leq 2 k q(\ell_*) + \sum_{A_1,\ldots,A_k}  \prod_{i=1}^k  \left( q(A_i) \cdot \delta_{|A_i| \geq \ell_*} \right) \cdot k(k-1) e^{-S^*_\infty(\ell_*)}  \cdot \sum_{x_1,\ldots,x_k} \prod_{i=1}^k p_\rho(x_i;A_i),
	\end{equation}
	where $S^*_\infty(\ell_*) = \min_{A:|A|\geq \ell_*} S^*_\infty(A) = \min_{A:|A|\geq \ell_*}\min_{y \in \text{sup}[\rho^A]}[S_\infty(\rho^A_y)]$ is the minimum post-measurement min-entropy after the state $\rho$ is projected onto any state $\ket{y}$ on any subregion $\bar A$ of size at most $n - \ell_*$ (see Section \ref{app:entropy}).
	We can then perform the summation over $x_1,\ldots,x_k$ and $A_1,\ldots,A_k$ to find,
	\begin{equation}
		\|\tilde{p}^\chi-\tilde{p}^\rho\|_{\rm TVD} \leq 2 k q(\ell_*) +  k(k-1) e^{-S^*_\infty(\ell_*)} \leq k e^{-\mathcal{O}(\gamma n)} +  k^2 e^{-S^*_k(\gamma n/4)} ,
	\end{equation}
	as claimed.
	
	The only remaining step is to confirm that $q(\ell_*) = e^{-\mathcal{O}(\gamma n)}$. 
	To show this, note that the size of the subregion $A$ is distributed according to a binomial distribution, $B(\ell;n,\gamma)=\binom{n}{\ell}\gamma^\ell(1-\gamma)^{n-\ell}$. Let us abbreviate $b_\ell \equiv B(\ell;n,\gamma)$ for convenience.
	We also let $\rho_\ell \equiv b_{\ell-1}/b_\ell =\frac{\ell}{\,n-\ell+1\,}\frac{1-\gamma}{\gamma}$ denote the ratio between adjacent values of $b_\ell$; note that $\rho_\ell$ is strictly increasing for $\ell < n/2$ and has value at most $1/2$ for $\ell \leq \gamma n / 4$.
	We can bound our desired probability as,
	\begin{equation}
		q(\ell_*) = \sum_{\ell=0}^{\ell_*-1} b_\ell \leq \sum_{\ell=0}^{\ell_*} b_\ell = \sum_{t=0}^{\ell_*} b_{\ell_*-t} = b_{\ell_*} \sum_{t=0}^{\ell_*} \rho_{\ell_*-t+1}\rho_{\ell_*-t+2}\cdots\rho_{\ell_*} \leq b_{\ell_*} \sum_{t=0}^{\ell_*} \rho_{\ell_*}^t \leq \frac{b_{\ell_*}}{1-\rho_k} \leq 2 b_{\ell_*},
		\label{eq:cutoff prob}
	\end{equation}
	where in the fifth expression we re-express $b_{\ell_*-t}$ using the definition of $\rho_\ell$, in the sixth expression we sum the geometric series, and in the seventh expression we apply $\rho_{\ell_*} < 1/2$ for $\ell_* = \gamma n/4$.
	To bound the remaining quantity $b_{\ell_*}$, we apply the standard Stirling inequality, $b_{\ell_*} \leq \frac{1}{\sqrt{2\pi \ell_*(1-\ell_*/n)}} \exp(-\ell_* \log(\frac{\ell_*}{n\gamma}) - (n-\ell_*)\log(\frac{n-\ell_*}{n-\gamma n}))$.
	Setting $\ell_* = \gamma n/4$ yields
	$b_{\gamma n/4} = \exp(-\mathcal{O}(\gamma n))$ as desired. \qed
	
	\section{Output distributions and universal fluctuations}\label{app:output distributions}
	
	In this section, we apply our approximate formula for the moments of the Scrooge ensemble (Section~\ref{app:approximate}) to characterize the measurement output distributions of states drawn from approximate Scrooge $k$-designs.
	We begin by providing straightforward calculations of the moments of individual output probabilities $p^\psi(x) = |\braket{x|\psi}|^2$. Using our tight error bounds, we confirm prior predictions that these moments should converge to those of re-scaled Porter-Thomas distributions~\cite{mark2024maximum}.
	%
	%
	%
	We then characterize the joint moments of $m$ distinct outcome probabilities $p^\psi(x_1)$,\ldots,  $p^\psi(x_m)$. For $m=2$ we prove that these converge to the (diagonal) moments of the $m \times m = 2\times 2$ Wishart distribution \cite{wishart1928generalised}, and we argue that this continues to hold for higher $m>2$. This distribution generically features non-negligible correlations between different bitstring probabilities, contrasting sharply with the behavior of Haar-random states, where the distribution of distinct bitstring probabilities are statistically independent.
	Finally, we present a simple method to experimentally characterize the emergence of Wishart statistics, by measuring the sensitivity of the output distribution to small amounts of read-out noise.
	Overall, the Scrooge ensemble predicts an intriguing universal relation between the noise sensitivity of computational (i.e.~$Z$-basis) measurements and \emph{off-diagonal} (i.e.~$X$-basis or $Y$-basis) quantum coherences of the background state $\rho$.
	%
	
	
	\subsection{1-bitstring marginals and the Porter-Thomas distribution}
	
	Let us first consider the 1-bitstring marginal of the distribution of outcome probabilities.
	For a Haar-random state, this yields the Porter-Thomas (PT) distribution,
	\begin{equation}
		\text{Pr}_{\text{PT}}\big( p(x) = p \big) = d \cdot e^{- d p }
	\end{equation}
	For a Scrooge-random state, existing works \cite{mark2024maximum} suggest that a \emph{re-scaled} Porter-Thomas (rPT) distribution emerges,
	\begin{equation}
		\text{Pr}_{\text{rPT}}\big( p(x) = p \big) = \frac{1}{p^\rho(x)} \cdot e^{- p / p^\rho(x) },
	\end{equation}
	with $p^\rho(x)$ replacing $1/d$. This being said, the analysis in these works is somewhat loose due to a lack of control over the relative error of the moments, which is particularly important given  that these quantities are exponentially small. Furthermore, the bounds on errors feature an unfavorable factor of $k!$.

	In this subsection, we use our tight relative error approximation for the Scrooge moments to confirm that all sub-exponential moments of the re-scaled Porter-Thomas distribution indeed emerge in the Scrooge ensemble, up to corrections that are exponentially small in multiplicative error.
	The moments of the rPT ensemble are easily computed,
	\begin{equation}
		\E_{p(x) \sim \text{rPT}} [p(x)^k] = \int dp \, \frac{1}{p^\rho(x)} \cdot e^{-p/p^\rho(x)} \cdot p^k = k! \cdot p^\rho(x)^k,
	\end{equation}
	Meanwhile, our relative error approximation formula for the Scrooge ensemble yields,
	\begin{equation}
		\E_{\psi \sim \mathcal{S}_\rho} p^\psi(x)^k = k! \cdot p^\rho(x)^k \cdot \left(1 + \mathcal{O}\big(k^2 2^{-S_\infty(\rho)} \big) \right),
		\label{eq:PTMoments}
	\end{equation}
	for any $k \leq e^{S_\infty(\rho)/3}$ [as stated in the main text as Eq.~\eqref{eq:porterthomas}].
	In the typical case where $S_\infty(\rho)$ is extensive in $n$, the $k$th moments closely agree for any $k$ that is sub-exponential in $n$.

	\subsection{2-bitstring marginals and the $2 \times 2$ Wishart distribution}
	
	Let us now turn to the 2-bitstring marginals of the distribution of outcome probabilities.
	Here, we will find that the behavior of the Scrooge ensemble is more subtle.
	For a Haar-random state, the 2-bitstring marginals are described by a tensor product of two independent Porter-Thomas distributions,
	\begin{equation}
		\text{Pr}_{\text{PT}^2}\left( \begin{pmatrix}
			p(x) \\
			p(x')
		\end{pmatrix}
		=
		\begin{pmatrix}
			p \\
			p'
		\end{pmatrix}
		\right) = d^2 \cdot e^{- d (p+p') }.
	\end{equation}
	This implies that the two probabilities are entirely uncorrelated with one another.
	In this subsection, we find that for a Scrooge-random state, the 2-bitstring marginals are instead described by a \emph{bivariate gamma} (BG) distribution~\cite{kibble1941two,kuriki2010graph},
	\begin{equation} \nonumber
		\text{Pr}_{\text{BG}}\left( \begin{pmatrix}
			p(x) \\
			p(x')
		\end{pmatrix}
		=
		\begin{pmatrix}
			p \\
			p'
		\end{pmatrix}
		\right) 
		= \frac{1}{p^\rho(x) p^\rho(x')}  \frac{1}{1-|r|^2}  \exp(-\frac{1}{1-|r|^2} \left(\frac{p}{p^\rho(x)}+\frac{p'}{p^\rho(x')} \right)) I_0\left( \frac{2|r|}{1-|r|^2} \sqrt{\frac{p}{p^\rho(x)} \frac{p'}{p^\rho(x')}} \right),
	\end{equation}
	where $I_0$ is the modified Bessel function of the first kind. The degree of correlation between $p(x)$ and $p(x')$ is controlled by the magnitude of a complex parameter,
	\begin{equation}
		r \coloneqq \frac{\bra{x} \rho \ket{x'}}{\sqrt{p_\rho(x) p_\rho(x')}},
	\end{equation}
	which is determined by the off-diagonal matrix elements of $\rho$ between the computational basis states $x$ and $x'$.
	We have $0 \leq |r| \leq 1$.
	When $r = 0$ (i.e.~$\rho$ has no off-diagonal matrix elements), the bivariate gamma distribution reduces to the tensor product of two re-scaled Porter-Thomas distributions.
	However, for $|r| > 0$ they strongly differ.
	The bivariate gamma distribution is most easily viewed as a marginal distribution of the diagonal elements of a random matrix $\bs{W}$ that itself has the \emph{complex central $2\times 2$ Wishart distribution} (with $\nu = 1$ samples)~\cite{wishart1928generalised,kibble1941two,goodman1963statistical,kuriki2010graph},
	\begin{equation}
		\text{Pr}_{\text{CW}} \!\left( \begin{pmatrix}
			p(x) & p_{x,x'} \\
			\bar{p}_{x,x'} & p(x')
		\end{pmatrix}
		=
		\bs{W}
		\right) 
		=  \frac{1}{\pi} \frac{1}{\text{det}(\bs{W})} \frac{1}{\text{det}(\bs{\Sigma})} e^{- \tr( \bs{\Sigma}^{-1} \bs{W} )},  \quad \text{ with } \quad  \bs{\Sigma} = \begin{pmatrix} p^\rho(x) & r\sqrt{p^\rho(x)p^\rho(x')} \\ \bar{r} \sqrt{p^\rho(x)p^\rho(x')} & p^\rho(x') \end{pmatrix} 
	\end{equation}
	where the domain of $\bs{W}$ is taken to be over all positive semi-definite matrices. The above probability density function describes a joint distribution over the real variables $p(x), p(x')$ and the complex variable $p_{x,x'}$.
	Here, the additional variable $p_{x,x'}$ should be interpreted as the off-diagonal matrix element of a Scrooge-random state, $\braket{x|\psi}\braket{\psi|x'}$.
	
	As in the previous subsection, we support these claims by comparing the moments of the output distributions of the Scrooge ensemble to the moments of the distributions above.
	The moments of the re-scaled complex central Wishart distribution are known~\cite{wishart1928generalised,kibble1941two,kuriki2010graph},
	\begin{equation}
		\E_{\text{rCW}} \left[ |p_{x,x'}|^{2a} p(x)^b p(x')^c \right] = p_\rho(x)^{b+a} p_\rho(x')^{c+a} \cdot (b+a)! (c+a)! \sum_{\ell=0}^{\min(b+a,c+a)} {b+a \choose \ell} {c+a \choose \ell} |r|^{2\ell}.
	\end{equation}
	Meanwhile, the moments of the tensor product of two Porter-Thomas distributions are obtained by setting $r \rightarrow 0$,
	\begin{equation}
		\E_{\text{PT}^2} \left[ p(x)^b p(x')^c \right] = p_\rho(x)^{b} p_\rho(x')^{c} \cdot b! \, c!.
	\end{equation}
	The moments of the Scrooge ensemble can be computed using our relative error formula,
	\begin{equation}
		\begin{split}
			\E_{\psi \sim \mathcal{S}_\rho}\left[ \big| \braket{x|\psi}\!\braket{\psi|x'}\big|^{2a} p^\psi(x)^b p^\psi(x')^c \right] 
			& \approx \sum_\pi \tr( \left(\dyad{x}{x'}^{\otimes a} \otimes \dyad{x'}{x}^{\otimes a} \otimes \dyad{x}^{\otimes b} \otimes \dyad{x'}^{\otimes c}\right) \cdot \rho^{\otimes (2a+b+c)} \cdot \pi )\\
			& = \sum_\pi \tr( \left(\dyad{x}^{\otimes (b+a)} \otimes \dyad{x'}^{\otimes (c+a)}\right) \cdot \rho^{\otimes (2a+b+c)} \cdot \pi )\\
			& = 
			p^\rho(x)^{b+a} p^\rho(x')^{c+a} \cdot (b+a)! (c+a)! \sum_{\ell = 0}^{\min(b+a,c+a)} {b+a \choose \ell} {c+a \choose \ell} |r|^{2\ell},
		\end{split}
	\end{equation}
	where the approximation in the first equality holds up to relative error $\mathcal{O}((2a+b+c)^{2}2^{-S_\infty(\rho)/2})$ for $2a+b+c \leq e^{S_\infty/3}$.
	In going from the first to second line, we re-index the summation over $\pi$ to exchange $a$ copies of $x'$ between $\dyad{x}{x'}^{\otimes a}$ and $\dyad{x'}{x}^{\otimes a}$.
	In going from the second to third line, we note that each term in the sum is given by $p^\rho(x)^{b+a-\ell} p^\rho(x)^{b+a-\ell} |\bra{x} \rho \ket{x'}|^{2\ell} = p^\rho(x)^{b+a} p^\rho(x)^{b+a} |r|^{2\ell}$, where $\ell$ counts the number of times that the permutation $\pi$ that exchanges $x$ with $x'$.
	There are $(b+a)!(c+a)! {b+a \choose \ell} {c+a \choose \ell}$ such $\pi$ for each value of $\ell$; summing over $\ell$ yields the third expression.
	This is precisely equal to the moments of the re-scaled complex central Wishart distribution, as claimed.

	\subsection{$m$-bitstring marginals and the $m \times m$ Wishart distribution}
	
	Here we provide a non-rigorous argument that the joint distribution of $m$ distinct bitstring probabilities $\vec{p} = (p(x_1), \ldots, p(x_m))$ in Scrooge-random states should converge towards the diagonal marginal of the complex central Wishart distribution for an $m \times m$ matrix $\bs{W}$. Using the general expression for expectation values over Scrooge random states, and making the approximation that $r_\phi$ can be replaced with its mean value of unity, the moment generating function for $\vec{p}$ is expected to be close to
	\begin{align}
		\mathbb{E}_{\psi \sim \mathcal{S}_\rho} \left[e^{-\sum_{i=1}^m t_i p(x_i)}\right] \approx \int \dif \mu_{\rm Haar}(\phi) \exp\left(-d\sum_i t_i|\braket{\phi|\sqrt{\rho}|x_i}|^2 \right) 
	\end{align}
	While we do not have a precise characterization of how accurate this approximation is, it would seem appropriate in light of the strong sub-exponential concentration implied by Lemma \ref{lem:RConc}. Moreover, it is well known that for large dimension $d$,  Gaussian random vectors $\ket{z} \sim \mathcal{N}_{\mathbb{C}^d}(0,1/d)$ approximate the Haar measure well, since $\ket{z}$ is approximately normalized with high probability. On the basis of this approximation, we have
	\begin{align}
		\mathbb{E}_{\psi \sim \mathcal{S}_\rho} \left[e^{-\sum_{i=1}^m t_i p(x_i)}\right] \approx \int \frac{\dif z \dif \bar{z}}{(2\pi)^d}  \exp\left(-d\Braket{z|I+\sum_it_iP_i|z} \right) = \det\left(I+\sum_it_i P_i\right)^{-1}
	\end{align}
	where $P_i = \sqrt{\rho}\ket{x_i}\bra{x_i}\sqrt{\rho}$.
	
	On the other hand, consider a $m\times m$ random matrix $\bs{W}$ with a complex central Wishart distribution parametrized by a positive semi-definite matrix $\bs{\Sigma}$. This distribution can be constructed by taking a complex $m$-dimensional Gaussian random vector $\vec{x} \sim \mathcal{N}_{\mathbb{C}^m}(0,\bs{\Sigma})$ with zero mean and covariance matrix $\bs{\Sigma}$, and setting $\bs{W} = \vec{x} \vec{x}^T$. The moment generating function for the diagonal elements are \cite{kuriki2010graph}
	\begin{align}
		\mathbb{E}_{\bs{W}} \left[e^{-\sum_{i=1}^m t_i W_{ii}}\right] &= \det(\bs{1} +  \bs{T}\bs{\Sigma})^{-1} & \text{where }\bs{T} = \text{diag}(t_1, \ldots, t_m)
	\end{align}
	We set
	\begin{align}
		[\bs{\Sigma}]_{ij} = \braket{x_i|\rho|x_j},
	\end{align}
	and then note that
	\begin{align}
		\sum_i t_i P_i &= VV^\dagger \text{ and } \bs{T}^{1/2} \bs{\Sigma}\bs{T}^{1/2} = V^\dagger V & \text{where }V = \sum_i \sqrt{t_i \rho}\ket{x_i}\bra{i}.
	\end{align}
	Using the fact that $\bs{T}^{1/2} \bs{\Sigma}\bs{T}^{1/2}$ and $\bs{T} \bs{\Sigma}$ are related by a similarity transformation, along with the Weinstein-Aronszajn identity $\det(\bs{1}+\bs{A}\bs{B}) = \det(\bs{1}+\bs{B}\bs{A})$, we see that the two moment generating functions coincide.
	
	\subsection{Universal relation between read-out noise sensitivity and off-diagonal quantum coherence}
	
	How can one test that the output distribution in the computational basis is sensitive to off-diagonal coherences of the density matrix $\rho$?
	The simplest signature is the moment,
	\begin{equation} \label{eq: 2nd moment noise}
		\E_{\psi \sim \mathcal{S}_\rho} \left[ p^\psi(x) p^\psi(x') \right] \approx p^\rho(x) p^\rho(x') + \big| \! \bra{x} \rho \ket{x'} \! \big|^2 =  p^\rho(x) p^\rho(x') \cdot (1+r),
	\end{equation}
	for two bitstrings $x$ and $x'$.
	For example, if $\rho$ contained a large expectation value of a Pauli $X$ operator on qubit $i$, then one could set $x'$ to be equal to $x$ on all bits except qubit $i$.
	In that case, let us denote the resulting bitstring as $x' \rightarrow \hat{x}_i$.
	This would yield a large matrix element $| \bra{x} \rho \ket{\hat{x}_i}|$, which would manifest in large deviations of the above quantity from its tensor-product Porter-Thomas value, $p^\rho(x) p^\rho(\hat{x}_i)$.
	To quantify this precisely, let
	\begin{equation}
		\rho_{x_{\neq i}} = \frac{\tr_{\neq i}( \dyad{x_{\neq i}} \rho) }{\tr( \dyad{x_{\neq i}} \rho )}
	\end{equation}
	denote the post-measurement state on qubit $i$ when all qubits except $i$ are measured in the state $x_{\neq i}$, where $x_{\neq i}$ is the restriction of $x$ or $\hat{x}_i$ to qubits other than $i$.
	We can then compute the value of $r$,
	\begin{equation}
		r = \frac{\big| \! \bra{0} \rho_{x_{\neq i}} \ket{1} \! \big|^2}{\bra{0} \rho_{x_{\neq i}} \ket{0} \bra{1} \rho_{x_{\neq i}} \ket{1}} = \frac{\alpha_{x_\neq i} \sin^2(\theta_{x_{\neq i}})}{1-\alpha_{x_{\neq i}} \cos^2(\theta_{x_{\neq i}})},
	\end{equation}
	where $\alpha_{x_{\neq i}} = 2 \text{tr}(\rho_{x_{\neq i}}^2) - 1 \leq 1$ quantifies the post-measurement purity of the state $\rho_{x_{\neq i}}$, and the variable $\theta_{x_{\neq i}} = \text{arccos}( \text{tr}( \rho_{x_{\neq i}} Z_i ) / \sqrt{\alpha_{x_{\neq i}}})$ measures its polar angle. 
	For any $0 < \theta_{x_{\neq i}} < \pi$, the value of $r$ is non-zero.

	A related signature is the sensitivity of the output distribution to readout noise, which we propose to probe using the collision probability $\sum_x p^\psi(x)^2$. Absent noise, the emergence of the rescaled Porter-Thomas distribution implies that the average collision probability is twice that of the background state, $\mathbb{E}_{\psi \sim \mathcal{S}_\rho} \sum_x p^\psi(x)^2 \approx 2\sum_x p^\rho(x)^2$, up to small relative error $\mathcal{O}(k^2 2^{-S_\infty(\rho)/2})$. 
	Suppose that the output distribution is acted on by readout noise of strength $\gamma$, or equivalently that depolarizing noise is applied to the state $\psi$ immediately before measurement, such that $p^\psi(x) \rightarrow \tilde{p}^\psi(x) = \braket{x|\mathcal{D}_\gamma[\psi]|x}$. This has the effect of randomizing each bit of $x$ with probability $\gamma$.
	At leading order in $\gamma$, this causes the collision probability of the output distribution to decay as,
	\begin{equation}
		\E_{\psi \sim \mathcal{S}_\rho} \sum_x \tilde{p}^\psi(x)^2  \approx \E_{\psi \sim \mathcal{S}_\rho} \sum_xp^\psi(x)^2 - \gamma \E_{\psi \sim \mathcal{S}_\rho}  \sum_i \sum_x p^\psi(x) p^\psi(\hat{x}_i) + \mathcal{O}(\gamma^2), 
	\end{equation}
	where, as in the previous paragraph, $\hat{x}_i$ is the bitstring $x$ with the $i$-th bit flipped.
	The $\mathcal{O}(\gamma)$ term can be computed from Eq.~(\ref{eq: 2nd moment noise}),
	\begin{equation}
		\gamma \E_{\psi \sim \mathcal{S}_\rho}  \sum_i \sum_x p^\psi(x) p^\psi(\hat{x}_i) \approx \gamma \sum_i \sum_x p^\rho(x) p^\rho(\hat{x}_i) + \gamma \sum_i \sum_x \big| \! \bra{x} \rho \ket{x'} \! \big|^2.
	\end{equation}
	The first term is the noise sensitivity of the output distribution of $\rho$, which can in principle be determined independently.
	Any increase in the noise sensitivity of $\chi$ compared to the noise sensitivity of $\rho$ is due to statistical correlations between output probabilities of distinct bitstrings, which serves as a signature of large off-diagonal matrix elements in $\rho$.

	\section{Sampling from Scrooge $k$-designs}

	In this section, we analyse the structure of sampling problems where the underlying states $\psi$ are drawn from approximate Scrooge $k$-designs with respect to a background state $\rho$. These results substantiate the discussion in the main text regarding the putative classical hardness (in the absence of noise) and easiness (in the presence of noise) of simulating sampling from such states. These calculations are motivated on the predication that the background state $\rho$ is sufficiently mixed so as to make sapling from $p^\rho(x)$ classically tractable. Our results highlight the key differences between sampling from $\rho$ and $\psi$ (or, in the noisy case, the lack thereof) that would appear to prohibit classical simulations of the latter.

	Specifically, we begin by showing that the output distributions of states drawn from Scrooge $4$-designs are far in TVD from the output distribution of the background state $\rho$. In contrast, when a small number of qubits are traced out, or noise is applied, then the two distributions become indistinguishable, and hence any algorithm that samples from $\rho$ can be used to simulate sampling from the noisy or marginal distribution of $\psi$.
	We then present rigorous computations of the classical conditional mutual information (CMI) of the output distributions of Scrooge-random states, and states drawn from Scrooge $k$-designs, showing that the output distribution exhibits CMI among any three subsets of qubits, irrespective of geometric locality. This rules out the usage of classical algorithms for sampling that rely on the short-ranged nature of CMI.
	%
	%
	
	\subsection{Total variation distance bounds on the output distributions of Scrooge-random states}
	
	In this subsection, we provide upper and lower bounds on the total variation distance between the output distributions of a Scrooge-random state $\psi$ in comparison to the output distribution of the background state $\rho$. 
	Our main results follow from our relative error bounds on the moments of the Scrooge ensemble in Sections~\ref{app:approximate} and~\ref{app:concentration}. 
	
	We first show that the expected total variation distance between the full output distribution of a Scrooge-random state and that of $\rho$ is large.
	\begin{corollary}[Scrooge-random output distributions are far from $\rho$] \label{cor:TVDfar}
		Let $\psi$ be drawn from an  $\epsilon$-approximate Scrooge $4$-design with relative error. The expected total variation distance between the output distribution  $p^\psi(x) = \braket{x|\psi|x}$ and the output distribution of the background state $p^\rho(x) = \braket{x|\rho|x}$ is at least
		\begin{align}
			\mathbb{E}_\psi \|p^\psi - p^\rho\|_{\rm TVD} \geq \frac{1}{3} - \mathcal{O}(2^{-S_\infty(\rho)/2} + \epsilon)
		\end{align}
	\end{corollary}
	\noindent This shows that for a typical, known instance of $\psi$, one can can always distinguish its output distribution from that of $\rho$ by a statistical test. Hence, one cannot spoof the output distribution of $\psi$ simply by sampling from $\rho$. The test is allowed to use knowledge of the particular instance $\psi$ that was drawn---for example, it could involve computing a cross-entropy between samples from the distribution being tested and the output distribution that one would obtain from $\psi$. This setting (which also applies to the subsequent two results) should be contrasted with that considered in Corollary \ref{cor: noise}, where the distinguisher receives no information about the particular instance $\psi$, but instead can perform adaptive sequences of measurements on successive copies of $\psi$.  
	
	Second, using our approximate formula for the Scrooge moments (Theorem \ref{thm:approx}), we prove a converse bound for marginal distributions. 
	\begin{corollary}[Marginal output distributions of Scrooge-random states are close to $\rho$] \label{cor:TVDcloseSubsys}
		Let $\psi$ be drawn from an  $\epsilon$-approximate Scrooge $2$-design with relative error, and let $A$ be a subset of qubits with complement $B$. The expected total variation distance between the marginal output distribution of $\psi$, $p^\psi_A(x_A) = \braket{x_A|\tr_B\psi|x_A}$, and the output distribution of the background state $p^\rho_A(x) = \braket{x_A|\tr_B\rho|x_A}$, is at most
		\begin{align}
			\delta_A^{\rm TVD} \coloneqq \mathbb{E}_\psi \|p^\psi_A - p^\rho_A\|_{\rm TVD} \leq \mathcal{O}(2^{-\frac{1}{2}S^*_\infty(B)} + \sqrt{\epsilon})
			\label{eq:TVDMarginal}
		\end{align}
		where $S^*_\infty(B)$ is the minimum post-measurement min-entropy, defined in \eqref{eq:post meas def}.
	\end{corollary}
	Finally, we show a closely related result, that the expected total variation distance between the output distribution of a noisy Scrooge-random state and that of $\rho$ is exponentially small in the noise rate times an entropic quantity which is typically extensive in system size.
	\begin{corollary}[Output distributions of noisy Scrooge-random states are close to $\rho$] \label{cor:TVDclose}
		Let $\psi$ be drawn from an  $\epsilon$-approximate Scrooge $2$-design with relative error, and let $\tilde{p}^\psi(x) = \braket{x|\mathcal{D}_\gamma[\psi]|x}$ and $\tilde{p}^\rho(x) = \braket{x|\mathcal{D}_\gamma[\rho]|x}$ denote the measurement outcome distributions after the depolarising channel of strength $\gamma$ has been applied to each qubit. The expected total variation distance between the two distributions is at most
		\begin{align}
			\mathbb{E}_{\psi} \|\tilde{p}^\psi - \tilde{p}^\rho\|_{\rm TVD} \leq e^{-\mathcal{O}(\gamma n)} + \mathcal{O}(2^{-\frac{1}{2}S_\infty^*(\gamma n/4)} + \sqrt{\epsilon})
		\end{align}
		where $S_\infty^*(\ell) \coloneqq \min_{A:|A| \geq \ell} S^*_\infty(A)$ is the minimum post-measurement min-entropy over all regions of size at least $\ell$.
	\end{corollary}
	\noindent This shows that, once a small amount of noise is applied, the output distribution of a typical element of a Scrooge $k$-design state cannot be distinguished from that of the background state $\rho$ by any statistical test. This holds even if one knows which instance $\psi$ was drawn from the ensemble. Hence, one can always spoof the output distribution of a noisy Scrooge-random state by simply sampling from the output distribution of $\rho$.
	
	
	At a high-level, these results mirror known properties of spherical $k$-designs produced by random quantum circuits, which has as its background state the maximally mixed distribution~\cite{boixo2018characterizing, shaw2025experimental}.
	Those results are substantially easier to derive in the Haar setting, due to the simplicity of the background distribution.
	To extend these results to Scrooge-random states, we require two technical insights.
	First, we use our strong relative-error bounds on the moments of the Scrooge ensemble in Sections~\ref{app:approximate} and~\ref{app:concentration}.
	Second, when bounding the total variation distance using higher moments of the output distribution, we ensure that the higher moment bounds always feature exactly one more copy of $p^\psi(x)$ or $p^\rho(x)$ in the numerator than in the denominator.
	This ensures that, when averaging over $\psi$, no terms with multiple copies of $\rho$ will appear, such as $\sum_x p^\rho(x)^2$ or $\sum_x p^\rho(x)^4$, which are difficult to treat.
	We speculate that this approach may be useful in other noisy circuit sampling settings, for example in random circuits with non-unital noise~\cite{fefferman2023effect,mele2024noise} or constant-depth random circuits~\cite{napp2022efficient,bao2024finite,mcginley2025measurement,bene2025quantum}, where challenges in bounding the total variation distance using second moments are also encountered.

	\subsection{Long-ranged conditional mutual information}
	
	Here we consider the conditional mutual information of the output distribution. Throughout, we will consider any tripartition of the system into subsets of qubits $A, B, C$, such that any bitstring can be specified as a triple of substrings $x = (x_A, x_B, x_C)$. Letting $X = (X_A, X_B, X_C)$ be a random bistring-valued variable with distribution $\text{Pr}(X=x) =p(x_A, x_B, x_C)$, the conditional mutual information (CMI) is
	\begin{align}
		I(X_A:X_C|X_B)_p = H(X_AX_B)_p+H(X_BX_C)_p - H(X_B)_p - H(X_AX_BX_C)_p,
		\label{eq:CMIDef}
	\end{align}
	where $H(Y) = -\sum_y \text{Pr}(Y=y)\log[\text{Pr}(Y=y)]$ is the Shannon entropy of an arbitrary discrete distribution. The CMI quantifies the average amount of correlations between the variables $X_A$ and $X_C$ (as quantified by the mutual information) after the variable $X_B$ is revealed.
	
	First, for states $\psi$ drawn from the exact Scrooge ensemble, an exact closed-form expression for the average CMI can be obtained in the special case where $\rho^{ABC}$ factorizes across the tripartition.
	\begin{theorem}[Output distributions of Scrooge designs have quantized CMI for product $\rho$] \label{thm:CMI quantized}
		For any tripartition $ABC$, let $\psi \sim \mathcal{S}_\rho$ be a Scrooge-random state with respect to a background density operator that satisfies $\rho = \rho^A \otimes \rho^B \otimes \rho^C$. The average CMI \eqref{eq:CMIDef} of the output distribution of $p^\psi(x) = \braket{x|\psi|x}$ with respect to the tripartition $x  =(x_A, x_B, x_C)$ is
		\begin{align}
			\mathbb{E}_{\psi \sim \mathcal{S}_\rho} I(X_A:X_C|X_B)_{p^\psi} &= Q(\rho^A) + Q(\rho^C) - Q(\rho^{AC}) \label{eq:CMI Exact} \\
			&=  \frac{1-\gamma}{\ln 2} + \mathcal{O}\left(2^{-S_2(\rho^A)} + 2^{-S_2(\rho^C)} + 2^{-S_2(\rho^{AC})}\right) \approx 0.61\textup{ bits} \label{eq:CMI asymptote}
		\end{align}
		where $Q(\rho) = \sum_i \left(\prod_{j\neq i}  \frac{\lambda_j}{\lambda_j - \lambda_i}\right)\lambda_i \log \lambda_i$, with $\lambda_i$ the eigenvalues of $\rho$, is the subentropy \cite{jozsa1994lower, datta2014}.
	\end{theorem}
	\noindent The subentropy---a quantity that is closely related to the Scrooge ensemble \cite{jozsa1994lower}---converges towards the quantized value $(1-\gamma)/\ln 2$, with $\gamma$ the Euler-Mascheroni constant, when $\rho$ has large entropy (see Lemma \ref{lem:subentropy bounds} below), hence the quantization \eqref{eq:CMI asymptote}. The above result implies that for regions $A$ and $C$ of size $\Omega(\log(1/\epsilon))$, irrespective of the geometrical distance between $A$ and $C$ we would still expect to see a universal quantized value $ \approx 0.61 \text{ bits'}$ for the CMI, up to error $\epsilon$. As we discuss in the main text, the presence of non-negligible CMI among all tripartitions suggests that simulating sampling from Scrooge random states cannot be simulated by an efficient classical algorithm, even if sampling from the background state can can be classically simulated. 
	

	In contrast to the majority of our results, which are well-controlled approximations, Theorem \ref{thm:CMI quantized} gives us an exact expression for the CMI. Thus, to derive it, we require that the distribution of $\psi$ be exactly that of the Scrooge ensemble, and we also make an assumption that $\rho$ is a tensor product across the tripartition. Using our approximate formula for the Scrooge moments (Theorem \ref{thm:approx}), we can substantially generalize this to settings where $\rho$ is arbitrary, and $\psi$ drawn from an approximate Scrooge $k$-design, rather than the exact Scrooge ensemble $\mathcal{S}_\rho$.
	\begin{corollary}
		[Output distributions of Scrooge designs have long-ranged CMI] \label{cor: CMI}
		Let $\ket{\psi}$ be drawn from an ensemble $\mathcal{E}$ that forms an $\epsilon$-approximate Scrooge 4-design with relative error, and let $ABC$ be any tripartition of qubits. The average CMI \eqref{eq:CMIDef} of the output distribution of $p^\psi(x) = \braket{x|\psi|x}$ with respect to the tripartition $x  =(x_A, x_B, x_C)$ satisfies
		\begin{equation}
			\mathbb{E}_{\psi \sim \mathcal{E}}I(X_A:X_C|X_B) \geq \frac{1}{8 \ln 2} - \sqrt{\delta_{A}^{\rm TVD}} -\sqrt{\delta_{C}^{\rm TVD}} - \sqrt{\epsilon}.
		\end{equation}
		where $\delta_{A}^{\rm TVD} = O(2^{-S^*_\infty(A)})$ is the error from Corollary \ref{cor:TVDcloseSubsys}.
	\end{corollary}
	\noindent These results serve multiple purposes. 
	As is the case for random circuits, a high CMI rules out many simple classical algorithms for simulating sampling. Thus, these results provide starting evidence for the hardness of this task. 
	Second, we envision that the universal quantized value of the CMI could be useful as a numerical or experimental signature of Scrooge-random behavior. In particular, choosing the regions $X_A$ and $X_C$ to be small, such that $Q(\rho^{A,C})$ is nonzero but appreciably different from its Haar value, this could principle allow one to distinguish the Scrooge and Haar ensembles.

	As an aside, one should distinguish the above results from a particular instatiation of CMI studied in a similar context in Section V H of Ref.~\cite{mark2024maximum}. The quantity considered there involves an ensemble of states $\{(p_t,\ket{\psi_t})\}_{t}$ with labels $t$ serving as the random variable on which we condition---the CMI they consider thus reduces to the average of a mutual information over this ensemble. In contrast, the CMI here is computed for a fixed instance $\psi$, with the conditioning variable being a substring of the outcome $x_B$, and the averaging is performed afterwards.
	
	\subsection{Proofs of Corollaries~\ref{cor:TVDfar},~\ref{cor:TVDcloseSubsys}, and~\ref{cor:TVDclose}}
	
	\begin{proof}[Proof of Corollary~\ref{cor:TVDfar}]
		The expected total variation distance between the distributions of $\psi$ and $\rho$ can be written,
		\begin{equation}
			\E_{\psi \sim \mathcal{S}_\rho} \| p^\psi -  p^\rho\|_{\rm TVD}
			= \E_{\psi \sim \mathcal{S}_\rho} \sum_x \left| p^\psi(x) - p^\rho(x) \right|
			= \E_{\psi \sim \mathcal{S}_\rho} \E_{x \sim p^\rho} \left| p^\psi(x)/p^\rho(x) - 1 \right|,
			\label{eq:TVDRatio}
		\end{equation}
		where $\E_{x \sim p^\rho} f(x) \equiv \sum_x p^\rho(x) f(x)$.
		Following recent work on random quantum circuits~\cite{nietner2023average,schuster2024random}, we will apply the fourth moment method~\cite{berger1997fourth} to lower bound this value. 
		The fourth moment method is simply the inequality,
		\begin{equation}
			\E_X |X| \geq \frac{(\E_X X^2)^{3/2}}{(\E_X X^4)^{1/2}}
			\label{eq:FourthMoment}
		\end{equation}
		for any real-valued random variable $X$.
		Applied to our formula for the total variation distance, we have
		\begin{equation}
			\E_{\psi \sim \mathcal{S}_\rho} \E_{x \sim p^\rho} \left| p^\psi(x)/p^\rho(x) - 1 \right|
			\geq 
			\frac{
				\left( \E_{\psi \sim \mathcal{S}_\rho} \E_{x \sim p^\rho} \left( p^\psi(x)/p^\rho(x) - 1 \right)^2 \right)^{3/2}
			}
			{
				\left( \E_{\psi \sim \mathcal{S}_\rho} \E_{x \sim p^\rho} \left( p^\psi(x)/p^\rho(x) - 1 \right)^4 \right)^{1/2}
			}.
		\end{equation}
		The numerator and denominator can now be computed using our approximate formulas for the moments of Scrooge $k$-designs for $k \leq 4$. In particular, we have the Porter-Thomas moments derived in Eq.~\eqref{eq:PTMoments}, $\E_{\psi \sim \mathcal{S}_\rho} p^\psi(x)^k \approx  k! p^\rho(x)^k$ up to relative error $\delta_{\rho, k} + \epsilon$, where $\delta_{\rho, k} = O(k^2 2^{-S_\infty(\rho)})$ is the error in the approximate moment formula, defined in Theorem \ref{thm:approx}.
		By expanding out the bracket $( p^\psi(x)/p^\rho(x) - 1 )^2$ in powers of $p^\psi(x)$, the second moment can be bounded as
		\begin{equation}
			\E_{\psi \sim \mathcal{S}_\rho} \E_{x \sim p^\rho} \left( p^\psi(x)/p_\rho(x) - 1 \right)^2
			\geq 
			\E_{x \sim p_\rho} \big( 2 - 2 + 1 - \mathcal{O}(\delta_{\rho, 2}+\epsilon) \big) 
			=
			1 - \mathcal{O}(\delta_{\rho, 2}+\epsilon)
		\end{equation}
		For the fourth moment, expanding the bracket and using the approximate moments, we find
		\begin{equation}
			\E_{\psi \sim \mathcal{S}_\rho} \E_{x \sim p^\rho} \left( p^\psi(x)/p^\rho(x) - 1 \right)^4
			\leq 
			\E_{x \sim p_\rho} \big( 24 - 24 + 12 - 4 + 1  + \mathcal{O}(\delta_{\rho, 4}+\epsilon)\big)
			=
			9 + \mathcal{O}(\delta_{\rho, 4}+\epsilon).
		\end{equation}
		Inserting into our lower bound on the expected total variation distance and invoking the upper bound on $\delta_{\rho, k}$ from Theorem \ref{thm:approx}, we find
		\begin{equation}
			\E_{\psi \sim \mathcal{S}_\rho} \| p^\psi- p^\rho \|_{\rm TVD}
			\geq 
			\frac{1}{3} - \mathcal{O}(2^{-S_\infty(\rho)/2}+\epsilon).
		\end{equation}
		which completes our proof.
	\end{proof}
	
	\begin{proof}[Proof of Corollary~\ref{cor:TVDcloseSubsys}]
		Following the same logic as Eq.~\eqref{eq:TVDRatio}, we have
		\begin{equation}
			\E_{\psi \sim \mathcal{S}_\rho} \|p^\psi_A - p^\rho_A \|_{\rm TVD}
			= \E_{\psi \sim \mathcal{S}_\rho} \E_{x_A \sim p^\rho_A} \left| \tilde{p}^\psi(x_A)/\tilde{p}^\rho(x_A) - 1 \right|.
		\end{equation}
		By Jensen's inequality, $\mathbb{E}[|X|] \leq \mathbb{E}[X^2]^{1/2}$, we have
		\begin{align}
			\E_{\psi \sim \mathcal{S}_\rho} \|p^\psi_A - p^\rho_A \|_{\rm TVD} 
			&\leq 
			\left( \E_{\psi \sim \mathcal{S}_\rho} \E_{x_A \sim p_A^\rho} \left( \frac{p^\psi(x_A)}{p^\rho(x_A)} - 1 \right)^2 \right)^{1/2} \nonumber\\
			&= 
			\left( \E_{\psi \sim \mathcal{S}_\rho} \E_{x_A \sim \tilde{p}_A^\rho} \left( \frac{p^\psi(x_A)^2}{p^\rho(x_A)^2} - 2 \frac{p^\psi(x_A)}{p^\rho(x_A)} + 1 \right) \right)^{1/2}.
		\end{align}
		By Theorem \ref{thm:subsystem full}, we have $\mathbb{E}_{\psi} p^\psi_A(x_A)^2 \approx p^\rho_A(x_A)^2 $ up to relative error $\epsilon +\delta_{\rm prod}[A]$, where $\delta_{\rm prod}[A] = 2^{-S^*_\infty(B)+1}$ as per Eq.~\eqref{eq:epsilonProd}. We also have $\mathbb{E}_{\psi} p^{\psi}_A(x_A) \approx p^{\rho}_A(x_A)$ up to relative error $\epsilon$. Hence,
		\begin{align}
			\E_{\psi \sim \mathcal{S}_\rho} \|p^\psi_A - p^\rho_A \|_{\rm TVD} &\leq \sqrt{\delta_{\rm prod}[A] + 3\epsilon}
		\end{align}
		and since $\sqrt{a+b} \leq \sqrt{a} + \sqrt{b}$ for $a, b \geq 0$, this completes the proof.
	\end{proof}
	
	\begin{proof}[Proof of Corollary~\ref{cor:TVDclose}]
		As in the proof of Corollary \ref{cor: noise}, we expand $\mathcal{D}_\gamma$ as a convex combination of channels which act as the fully depolarising channel on subsets of qubits. Namely, $\mathcal{D}_\gamma[\,\cdot\,] = \sum_{A \subseteq [n]} q(A) \tr_{A}[\,\cdot\,] \mathbbm{1}^A/2^{|A|}$, where $q(A) = \gamma^{|A|}(1-\gamma)^{n-|A|}$. Then, using the triangle inequality for TVD $\|\sum_i \alpha_i(p_i - q_i)\|_{\rm TVD} \leq \sum_i \alpha_i \|p_i -q_i\|_{\rm TVD}$ for any $\alpha_i \geq 0$, we have
		\begin{align}
			\E_\psi \|\tilde{p}^\psi(x) - \tilde{p}^\rho(x)\|_{\rm TVD} \leq \E_\psi\sum_{A \subseteq [n]}q(A) \|p^\psi_A(x_A) - p^\rho_A(x_A)\|_{\rm TVD}
		\end{align}
		Again we separate out regions $A$ of size less than $\ell_* = n\gamma/4$, which have total weight $q(\ell_*) \coloneqq \sum_{A : |A|< \ell_*} q(A) \leq e^{-\mathcal{O}(\gamma n)}$ [see Eq.~\eqref{eq:cutoff prob} and the discussion thereafter]. Since the TVD for any two distributions is upper bounded by 2, we have
		\begin{align}
			\E_\psi \|\tilde{p}^\psi(x) - \tilde{p}^\rho(x)\|_{\rm TVD} &\leq q(\ell_*) + \sum_{\substack{A \subseteq [n] \\ |A| \geq \ell_*}} q(A) \|p^\psi_A(x_A) - p^\rho_A(x_A)\|_{\rm TVD} \nonumber\\
			&\leq e^{-\mathcal{O}(\gamma n)} + \max_{A : |A| \geq \ell_*}\delta^{\rm TVD}_A
		\end{align}
		where $\delta^{\rm TVD}_A$ is defined and bounded in Corollary \ref{cor:TVDcloseSubsys}. By the definition of $S^*_\infty(\ell_*)$, this completes the proof.
	\end{proof}

	\subsection{Proof of Theorem~\ref{cor: CMI}: Quantization of conditional mutual information for product $\rho$}
	
	Our starting point is the general formula for averages of arbitrary functions $f(\ket{\psi})$ over Scrooge random states, Eq.~\eqref{eq:Scrooge expectation}, which we restate here
	\begin{align}
		\mathbb{E}_{\psi \sim \mathcal{S}_\rho}\big[f(\ket{\psi})\big]  &= \int \dif \mu_{\rm Haar}(\phi) r_\phi  f\left( \frac{\sqrt{d \rho}\ket{\phi}}{\sqrt{r_\phi}} \right), & \text{where }r_\phi \coloneqq d\braket{\phi|\rho|\phi}.
	\end{align}
	Using the fact that complex $d$-dimensional Gaussian vectors $\ket{z} = \sum_i z_i \ket{i}$, with $z_i \sim \mathcal{N}_\mathbb{C}(0,1)$, can be written as a product of a normalization $s = \braket{z|z}$ with a Haar-random normalized vector $\ket{\phi_z} = \ket{z}/\sqrt{\braket{z|z}} \sim \text{Haar}(d)$, we can rewrite the above as
	\begin{align}
		\mathbb{E}_{\psi \sim \mathcal{S}_\rho}\big[f(\ket{\psi})\big]  &= \int \frac{\dif z \dif \bar{z}}{(2\pi)^d} e^{-\braket{z|z}} r_z f\left( \frac{\sqrt{\rho}\ket{z}}{\sqrt{r_z}}  \right) & \text{where }r_z \coloneqq \braket{z|\rho|z}
	\end{align}
	We will compute the CMI by evaluating each of the four Shannon entropies in Eq.~\eqref{eq:CMIDef}. Accordingly, we will fix an arbitrary bipartition $AB$ and choose $f(\ket{\psi}) = H(X_A)_{p^\psi} = -\sum_{x_A}p^\psi(x_A)\log p^\psi(x_A)$, where $p^\psi(x_A) = \braket{x_A|\tr_B\psi |x_A}$ as usual. We can write $r_zf(\sqrt{\rho}\ket{z}/\sqrt{r_z}) = -\sum_{x_A}\tilde{p}^z(x_A)[\log(\tilde{p}^z(x_A)) - \log(r_z)]$, where
	\begin{align}
		\tilde{p}^z(x_A) &\coloneqq \braket{z|P_{x_A}|z} & \text{where } P_{x_A} \coloneqq \rho^{1/2}(\ket{x}\bra{x}_A\otimes I^B)\rho^{1/2}.
	\end{align}
	Using the fact that $\sum_{x_A} \tilde{p}^z(x_A) = r_z$, we find
	\begin{align}
		\mathbb{E}_{\psi \sim \mathcal{S}_\rho}\big[f(\ket{\psi})\big] &= \int \frac{\dif z \dif \bar{z}}{(2\pi)^d} e^{-\braket{z|z}} r_z \log r_z -\sum_{x_A} \int \frac{\dif z \dif \bar{z}}{(2\pi)^d} e^{-\braket{z|z}} \tilde{p}^z(x_A)\log \tilde{p}^z(x_A) \nonumber\\
		&\eqqcolon \mathbb{E}_z[r_z\log r_z] - \sum_{x_A}\mathbb{E}_z[\tilde{p}^z(x_A)\log \tilde{p}^z(x_A)],
		\label{eq:quantized cmi split}
	\end{align}
	where we use $\mathbb{E}_z$ to denote the average over uniform Gaussian random vectors $\ket{z}$. 
	The moment generating function of $\tilde{p}^z(x_A)$ over Gaussian random vectors $\ket{z}$ can then be written
	\begin{align}
		I_{x_A}(t) \coloneqq \mathbb{E}_z[ e^{t \tilde{p}^z(x_A)}] = \int \frac{\dif z \dif \bar{z}}{(2\pi)^d} e^{-\braket{z|I-tP_{x_A}|z}} = \det(I-tP_{x_A})^{-1} = \det(I^B-t p^\rho(x_A)\rho^B)^{-1}.
	\end{align}
	In the last step, we invoke the property $\rho = \rho^A \otimes \rho^B$, which implies that the eigenvalues of $P_{x_A}$ are either zero or of the form $p^\rho(x_A)\times \lambda^B_i$, where $p^\rho(x_A) = \braket{x_A|\rho|x_A}$, and $\lambda^B_i$ are the eigenvalues of $\rho^B$. We now use the fact that the Laplace transform of $f(x) = x\ln x$ is $g(s) = s^{-2}(1-\gamma) - s^{-2}\ln s$, where $\gamma$ is the Euler-Mascheroni constant (see Section 1.6, Eq.~6.9 in Ref.~\cite{oberhettinger2012tables}). By taking the inverse Laplace transform, we can therefore conclude that for any positive-valued random variable $X > 0$ such that $\mathbb{E}[e^{c X}] < \infty$ for some $c > 0$, we have
	\begin{align}
		\mathbb{E}_{X}[X\log(X/a)] = \frac{1-\gamma}{\ln 2}\mathbb{E}_{X}[X]- \frac{a}{2\pi \iu}\int_{-\iu \infty+c}^{+\iu \infty + c} \dif s\, \mathbb{E}[e^{sX/a}] s^{-2}\log s 
	\end{align}
	for any constant $a > 0$. We recognize the expectation value inside the integral as the moment generating function evaluated at $t = s/a$. Thus, taking $a = p^\rho(x_A)$ and using $\mathbb{E}_z [\tilde{p}^z(x_A)] = p^\rho(x_A)$, we have
	\begin{align}
		\mathbb{E}_z[\tilde{p}^z(x_A)\log\tilde{p}^z(x_A)] &= p^\rho(x_A)\log p^\rho(x_A)+ \frac{1-\gamma}{\ln 2}p^\rho(x_A) - \frac{p^\rho(x_A)}{2\pi \iu}\int_{-\iu \infty+c}^{+\iu \infty + c} \dif s\,s^{-2}\log s \prod_{i}\left(1-s \lambda^B_{i}\right)^{-1}.
	\end{align}
	The poles of the integrand in the domain $\Re s > 0$ occur at $s = 1/\lambda^B_i$, and assuming all eigenvalues of $\rho^B$ are distinct, these are all simple poles, with corresponding residue $\lambda^B_i\log(\lambda_i^B)\prod_{j \neq i}(1-\lambda_i^B/\lambda_j^B)^{-1}$. Deforming the contour to enclose all these poles in the clockwise direction, we obtain
	\begin{align}
		-\sum_{x_A}\mathbb{E}_z[\tilde{p}^z(x_A)\log \tilde{p}^z(x_A)] &= H(p^\rho_A) - \frac{1-\gamma}{\ln 2} + \sum_i \left(\prod_{j\neq i} \frac{\lambda_i^B}{\lambda_i^B - \lambda_j^B}\right) \lambda^B_i \log\lambda^B_i \nonumber\\ &= H(p^\rho_A) - \frac{1 - \gamma}{\ln 2} +  Q(\rho^B)
	\end{align}
	where $Q(\rho^B)$ is the subentropy \cite{jozsa1994lower, datta2014}. Finally, we can evaluate the first term in \eqref{eq:quantized cmi split} by taking the limit of $A = \emptyset$ in the above equation, since $\tilde{p}^z(x_\emptyset) = r_z$, which gives
	\begin{align}
		\mathbb{E}_z[r_z \log r_z] = \frac{1-\gamma}{\ln 2} - Q(\rho) 
		\label{eq:rz subentropy}
	\end{align}
	where $\rho = \rho^A \otimes \rho^B$ is the global state. Therefore,
	\begin{align}
		\mathbb{E}_{\psi \sim \mathcal{S}_\rho} H(X_A)_{p^\psi} = H(p^\rho_A)+Q(\rho^B) - Q(\rho^A \otimes \rho^B).
	\end{align}
	Applying this formula for each of the four terms in \eqref{eq:CMIDef}, and using the product structure of $\rho$ to infer that $H(p^\rho_{AB}) = H(p^\rho_{A}) + H(p^\rho_{B})$ and similar, we obtain
	\begin{align}
		\mathbb{E}_{\psi \sim \mathcal{S}_\rho}  I(X_A:X_C|X_B)_{p^\psi} = Q(\rho^A) + Q(\rho^C) - Q(\rho^{AC}).
	\end{align}
	
	Finally, the convergence of the subentropies $Q(\rho^A)$, $Q(\rho^C)$ to the quantized value $\frac{1-\gamma}{\ln 2}$, which implies the final statement, Eq.~\eqref{eq:CMI asymptote}, follows from the following lemma, which may be of independent interest in light of Ref.~\cite{datta2014}.
	\begin{lemma} \label{lem:subentropy bounds}
		For any density matrix, the subentropy $Q(\rho)$ satisfies $\frac{1-\gamma}{\ln 2} - \delta \leq Q(\rho) \leq \frac{1-\gamma}{\ln 2}$, where $\delta \leq \tr[\rho^2] = 2^{-S_2(\rho)}$.
	\end{lemma}
	\noindent We note that the upper bound was established in Ref.~\cite{datta2014} using a proof technique different from ours, presented below. The above suffices to prove Eq.~\eqref{eq:CMI asymptote}.
	
	\textit{Proof of Lemma \ref{lem:subentropy bounds}} We start from Eq.~\eqref{eq:rz subentropy}, and our aim will be to show $ 0 \leq \mathbb{E}_z[r_z \log r_z] \leq \tr[\rho^2]$, which suffices to prove Lemma \ref{lem:subentropy bounds}. For this purpose, we use the series of inequalities $
	0 \leq \mathbb{E}[X\log X] \leq \log(\mathbb{E}[X^2])$, which hold for any non-negative random variable $X$ with unit mean $\mathbb{E}[X] = 1$. The first of these is simply Jensen's inequality for the convex function $x \mapsto x \log x$. To show the second, we take the probability density function $f(x)$, and define a new distribution with density function $g(x) = xf(x)$, which is properly normalized thanks to the assumption on $X$. We then have $\mathbb{E}[X\log X] = \int_0^\infty \dif x g(x) \log(\frac{g(x)}{f(x)}) = D_{KL}(g\| f)$, while $\log(\mathbb{E}[X^2]) = \log(\int_0^\infty \dif x\, g(x) \frac{g(x)}{f(x)}) = D_2(g\|f)$, where $D_\alpha(g\|f) = \frac{1}{\alpha-1}\log (\int\dif x\, g(x)^\alpha f(x)^{1-\alpha})$ is the $\alpha$-R{\'e}nyi divergence, which tends to $D_{KL}$ in the limit $\alpha \rightarrow 1$. The inequality then follows from the fact that $D_\alpha$ is non-decreasing in $\alpha$ \cite{vanErven2014renyi}.
	
	We have $\mathbb{E}[r_z] = \mathbb{E}_z[\braket{z|\rho|z}] = \tr[\rho] = 1$, and so we immediately have that the left hand side of \eqref{eq:rz subentropy} is non-negative, thus proving the lower bound. By standard Gaussian integration methods, $\mathbb{E}[r_z^2] = \tr[\rho]^2 + \tr[\rho^2]$, and by the inequality proved above, we have $\mathbb{E}_z[r_z \log r_z] \leq \log(1+\tr[\rho^2]) \leq \tr[\rho^2]$.
	
	\qed
	
	\subsection{Proof of Corollary~\ref{cor: CMI}: Long-range conditional mutual information for Scrooge $k$-designs}
	
	Since we are working with classical probability distributions, we can re-express the CMI as an average of the mutual information $I(X_A:X_C)_{p^\psi_{x_B}}$ of the conditional distributions $p^\psi_{x_B}(x_Ax_C) \coloneqq p^\psi(x_Ax_Bx_C)/p^\psi(x_B)$. Specifically,
	\begin{align}
		I(X_A:X_C|X_B)_{p^\psi} = \sum_{x_B} p^\psi(x_B) I(X_A:X_C)_{p^\psi_{x_B}} 
	\end{align}
	For each $x_B$, the mutual information $I(X_A:X_C)_{p^\psi_{x_B}} $ is itself equal to the Kullback-Liebler divergence $D_{KL}(p(z)\|q(z)) = \sum_z p(x) \log[p(z)/q(z)]$ between the true conditional distribution $p^\psi_{x_B}(x_Ax_C)$ and $ p^\psi_{x_B}(x_A)p^\psi_{x_B}(x_C)$, which is the product of the marginalised conditional distributions $p^\psi_{x_B}(x_A) \coloneqq \sum_{x_C} p^{\psi}_{x_B}(x_Ax_C)$ (and similar for $p^\psi_{x_B}(x_C)$). We will now use Pinsker's inequality
	\begin{align}
		D_{KL}(p\|q) \geq \frac{2}{\ln 2} \|p - q\|_{\rm TVD}^2,
	\end{align}
	where $\|p - q\|_{\rm TVD} \coloneqq \sum_z |p(z) - q(z)|$ is the total variation distance between two distributions. Combining this with Jensen's inequality $\mathbb{E}[X^2] \geq \mathbb{E}[X]^2$, we can lower bound the CMI as
	\begin{align}
		\mathbb{E}_{\psi \sim \mathcal{E}}I(X_A:X_C|X_B) &\geq \frac{2}{\ln 2}\Delta^2 & \text{where } \Delta \coloneqq \mathbb{E}_{\psi \sim \mathcal{E}} \mathbb{E}_{x_B \sim \psi} \sum_{x_Ax_C} |p^\psi_{x_B}(x_Ax_C) - p^{\psi}_{x_B}(x_A) p^{\psi}_{x_B}(x_C)|
	\end{align}
	where we use the shorthand $\mathbb{E}_{x_B \sim \psi}f(x_B) = \sum_{x_B} p^\psi(x_B) f(x_B)$. From Corollary \ref{cor:TVDcloseSubsys}, we can obtain the following intermediate result
	\begin{align}
		\mathbb{E}_\psi \mathbb{E}_{x_B \sim \psi} \sum_{x_A}  \left| p^\psi_{x_B}(x_A) - p^\rho_{x_B}(x_A) \right| &= \mathbb{E}_{\psi} \sum_{x_B} \left| p^\psi(x_Ax_B) - p^\psi(x_B)p^\rho_{x_B}(x_A) \right| \nonumber\\
		&\leq \mathbb{E}_{\psi} \sum_{x_Bx_A} \left| p^\psi(x_Ax_B) - p^\rho(x_Ax_B) \right| + \sum_{x_B} |p^{\psi}(x_B) - p^\rho(x_B)|\sum_{x_A}p^\rho_{x_B}(x_A)  \nonumber\\
		&\leq \sqrt{\delta_{C}^{\rm TVD}} + \sqrt{\delta_{AC}^{\rm TVD}}
	\end{align}
	where $\delta_{C}^{\rm TVD}$ is defined in Eq.~\eqref{eq:TVDMarginal}. By the reverse triangle inequality, we have
	\begin{align}
		\Delta &=  \mathbb{E}_{\psi \sim \mathcal{E}} \sum_{x_Ax_Bx_C} \bigg|\big(p^\psi(x_Ax_Bx_C) - p^{\rho}_{x_B}(x_A) p^{\rho}(x_Bx_C)\big) + p^\rho_{x_B}(x_A)\big(p^{\rho}(x_Bx_C)- p^{\psi}(x_Bx_C) \big)  \nonumber\\ & \hspace*{60pt}  + p^{\psi}(x_Bx_C) \big(p^{\rho}_{x_B}(x_A) - p^{\psi}_{x_B}(x_A) \big) \bigg| \nonumber\\
		&\geq \mathbb{E}_{\psi \sim \mathcal{E}} \left[\sum_{x_Ax_Bx_C} |p^\psi(x_Ax_Bx_C) - p^{\rho}_{x_B}(x_A) p^{\psi}(x_Bx_C)| -\sum_{x_A}| p^{\psi}_{x_B}(x_A) - p^{\rho}_{x_B}(x_A)| - \sum_{x_Bx_C}| p^{\psi}(x_Bx_C) - p^{\rho}(x_Bx_C)|\right] \nonumber\\
		&\geq  \sum_{x_Ax_Bx_C} \mathbb{E}_{\psi \sim \mathcal{E}}|p^\psi(x_Ax_Bx_C) - p^{\rho}_{x_B}(x_A) p^{\rho}(x_Bx_C)| - \mathcal{O}(\sqrt{\delta_A^{\rm TVD}} + \sqrt{\delta_C^{\rm TVD}})
	\end{align}
	In going to the second line, we use the reverse triangle inequality $|x+y| \geq |x| - |y|$, along with the fact that $p_{x_B}^\rho(x_A)$ and $p^\psi(x_Bx_C)$ are properly normalised sets of probabilities. We now lower bound each term in the sum using the fourth moment method \eqref{eq:FourthMoment}, which we restate here
	\begin{align}
		\mathbb{E}[|X|] \geq \frac{(\mathbb{E}[X^2])^{3/2}}{(\mathbb{E}[X^4])^{1/2}}.
	\end{align}
	The second moment in the numerator can be evaluated using the fact that $\mathcal{E}$ is a Scrooge 2-design with relative error $\epsilon$ moments of $p^{\psi}(x_Ax_Bx_C)$ are close to those of the Porter-Thomas distribution [Eq.~\eqref{eq:PTMoments}]
	\begin{align}
		\mathbb{E}_{\psi \sim \mathcal{E}} \big(p^\psi(x_Ax_Bx_C) - p^{\rho}_{x_B}(x_A) p^{\rho}(x_Bx_C)\big)^2 &\geq p^\rho(x_Ax_Bx_C)^2(1 - 2\epsilon) + \big(p^\rho(x_Ax_Bx_C) - p_{x_B}^\rho(x_A)p^\rho(x_Bx_C)\big)^2 \nonumber\\ &\equiv p^2(1 + x^2 - 2\epsilon)
	\end{align}
	For notational convenience, we are temporarily using the shorthand $p \coloneqq p^\rho(x_Ax_Bx_C)$, $q \coloneqq p^{\rho}_{x_B}(x_A) p^{\rho}(x_Bx_C)$, $x = q/p - 1$. The fourth moments can be evaluated similarly
	\begin{align}
		\mathbb{E}_{\psi \sim \mathcal{E}} \big(p^\psi(x_Ax_Bx_C) - p^{\rho}_{x_B}(x_A) p^{\rho}(x_Bx_C)\big)^4 &\leq 24p^4 - 24 p^3q + 12 p^2q^2 - 4pq^3 + q^4 + p^4\mathcal{O}(\epsilon) \nonumber\\
		&= p^4\left(9 - 8 x + 6x^2 + x^4 + \mathcal{O}(\epsilon)\right)  \nonumber\\
		&\leq 16p^4\big[(1 + x^2)^2 + \mathcal{O}(\epsilon)\big]
	\end{align}
	Since $S_\infty(\rho) \geq S^*_\infty(A)$, we have $\delta_{\rho, 4} = \mathcal{O}(\delta^{\rm TVD}_A)$, and so
	\begin{align}
		\Delta \geq \frac{1}{4} - \mathcal{O}\left(\epsilon + \sqrt{\delta_A^{\rm TVD}} + \sqrt{\delta_C^{\rm TVD}}\right)
	\end{align}
	and in turn
	\begin{align}
		\mathbb{E}_{\psi \sim \mathcal{E}} I(X_A:X_C|X_B) \geq \frac{1}{8 \ln 2} - \mathcal{O}\left(\epsilon + \sqrt{\delta_A^{\rm TVD}} + \sqrt{\delta_C^{\rm TVD}}\right)
	\end{align}
	as claimed in Corollary \ref{cor: CMI}. \hfill $\square$


	\section{Lower bounds on the formation of Scrooge-random states}\label{sec:Bits}
	
	In this section, we present the full details of our lower bounds on the formation of Scrooge-random states.
	We begin by providing the proof of our lower bound on the number of bits of randomness required to form a Scrooge state $k$-design, and then apply our lower bound to the temporal ensemble.
	
	\subsection{Lower bound on the bits of randomness of Scrooge state designs}
	
	Any ensemble generated by a finite number $m < \infty$ of bits of randomness must contain at most $2^m$ distinct states. We can bound the number of bits of randomness needed to form an approximate Scrooge $k$-design using the following result.
	\begin{lemma} \label{lem:DiscreteBound}
		Suppose that an ensemble of states $\mathcal{E}$ has a discrete distribution $\{(p_i, \ket{\psi_i})\}_{i=1}^r$ of cardinality $r$, and forms an $\epsilon$-approximate Scrooge $k$-design with additive error. Then
		\begin{align}
			r \geq \left(1-\delta_{\rho,k} -\frac{\epsilon}{2} \right)\frac{2^{k S_\infty(\rho)}}{k!}
			\label{eq:RankBoundDiscrete}
		\end{align}
		where $\delta_{\rho, k} = O(k^2 2^{-S_\infty(\rho)/2})$ is the approximation error in Theorem \ref{thm:approx}. 
	\end{lemma}
	\textit{Proof of Lemma \ref{lem:DiscreteBound}.---}The stated assumption implies that the $k$th moment of $\mathcal{E}$, $\chi_{\mathcal{E}}^{(k)} = \sum_i p_i \dyad{\psi_i}^{\otimes k}$, has rank at most $r$, and also that $\|\chi_{\mathcal{E}}^{(k)} - \chi_{\mathcal{S}_\rho}^{(k)}\|_1 \leq  \epsilon$. Converting Theorem \ref{thm:approx} from relative to additive error, this implies that $\|\chi_{\mathcal{E}}^{(k)} - \chi_{\mathcal{S}_\rho}^{(k), \text{Appr}}\|_1 \leq  \epsilon + 2\delta_\rho$. Now, we use the fact that for an arbitrary state $\sigma$ with $\text{rank}(\sigma) > r$, and an arbitrary rank-$r$ state $\tau_r$ supported by a rank-$r$ projector $\Pi_\tau$, we have
	\begin{align}
		\|\sigma - \tau_r\|_1  &= \sup_{U\in \mathrm{U}(d)}\tr[U(\sigma-\tau_r)] \nonumber\\
		&\geq \tr[(1-2\Pi_\tau)(\sigma-\tau)] \nonumber\\
		&= 2-2\tr[\Pi_\tau \sigma]  \nonumber\\ &\geq 2-2\|\Pi_\tau\|_1\|\sigma\|_\infty= 2 - 2r \|\sigma\|_\infty
	\end{align}
	The final line is due to the fact that for any projector $P$, we have $\|XP\|_1 \leq \|P\|_1\|X\|_\infty = \text{rank}(P)\|X\|_\infty$ by H{\"o}lder's inequality. This general result we apply with $\chi^{(k), \text{Appr}}_{\mathcal{S}_\rho}$ in place of $\sigma$, and $\chi_\mathcal{E}^{(k)}$ in place of $\tau_r$, which indeed has rank $r$. The maximum eigenvalue of the approximate moment is $k! 2^{-k S_\infty(\rho)}$, corresponding to the eigenvector $\ket{\lambda^\star}^{\otimes k}$, where $\ket{\lambda^\star}$ is the eigenvector of $\rho$ with maximum eigenvalue. Eq.~\eqref{eq:RankBoundDiscrete} then follows. \qed
	
	\subsection{Failure of subexponential-time temporal ensembles}
	
	As discussed in the main text, one interesting consequence of our lower bound is that the temporal ensemble cannot realize a Scrooge state $k$-design for any evolution times subexponential in the number of qubits $n$.
	\begin{corollary}[Failure of subexponential-time temporal ensembles]\label{cor:temporal}
		For any Hamiltonian $H$ normalized such that $\lVert H \rVert_\infty = n$ and any initial state $\ket{\psi_0}$, define the temporal ensemble $\mathcal{E}_T = \{ e^{-iHt} \ket{\psi_0} \, | \, t \sim \textup{Unif}(0,T) \}$. Then, if $\mathcal{E}_T$ forms an $\epsilon$-approximate Scrooge $k$-design with additive error, the range of evolution times $T$ must be at least
		\begin{align}
			T \geq \mathcal{O}\left( \frac{2^{kS_\infty(\rho)}}{k!} \frac{(1-\epsilon/2)^2}{n}\right).
		\end{align}
	\end{corollary}
	For typical mixed states where $S_\infty(\rho) = \Theta(n)$, this gives $T \geq 2^{\mathcal{O}\big(k(n-\log k)\big)}$, which is exponentially large in $nk$ for all subexponential $k$.
	
	\begin{proof}[Proof of Corollary \ref{cor:temporal}]
		Our proof follows immediately from combining our lower bound on the number of bits of randomness in Scrooge $k$-designs with additive error, and a recent upper bound on the number of bits of randomness contained in any temporal ensemble~\cite{cui2025random}. In particular, there it was shown that for any integer $m \in \mathbb{N}^+$, one can construct an ensemble of pure states $\mathcal{E}'_m=\{(p_i, \ket{\phi_i})\}_{i=1}^m$ of cardinality $m$ whose $k$th moments $\chi^{(k)}_{\mathcal{E}'_m}$ are $\eta_m$-close to those of the temporal ensemble $\chi^{(k)}_{\mathcal{E}_T}$ in trace distance, where $\eta_m \leq 2kT\|H\|_\infty /m$. By the reverse triangle inequality, we have
		\begin{align}
			\|\chi^{(k)}_{\mathcal{E}_T} - \chi^{(k)}_{\mathcal{S}_\rho}\|_1 &\geq \|\chi^{(k)}_{\mathcal{E}'_m} - \chi^{(k)}_{\mathcal{S}_\rho}\|_1 - \|\chi^{(k)}_{\mathcal{E}_T} - \chi^{(k)}_{\mathcal{E}'_m}\|_1 \nonumber\\
			&\geq 1 -m(k!)2^{-kS_\infty(\rho)} - 2\delta_{\rho, k} - 2k\|H\|_\infty T/m
		\end{align}
		where the second inequality uses Lemma \ref{lem:DiscreteBound}. Setting $m = \left\lfloor\sqrt{T \|H\|_\infty  2^{-S_\infty(\rho)/2}/(k-1)!}\right\rfloor$ gives $\|\chi^{(k)}_{\mathcal{E}_T} - \chi^{(k)}_{\mathcal{S}_\rho}\|_1 \geq 1 - 2\delta_{\rho, k} -\mathcal{O}(\sqrt{T\|H\|_\infty2^{-kS_\infty(\rho)} k!})$. This remains larger than $\epsilon$ provided $T \leq (1-\delta_{\rho, k} -\epsilon/2)^2 2^{kS_\infty(\rho)}/(k!\|H\|_\infty )$, and since $\delta_{\rho, k} \leq \mathcal{O}(2^{-S_\infty(\rho)})$ the claim then follows. 
	\end{proof}
	
	\subsection{Lower bound on state complexity of Scrooge-random states}
	
	Finally, as a consequence of Lemma \ref{lem:DiscreteBound}, we can put a lower bound on the state complexity of states drawn from a Scrooge $k$-design, which mirrors analogous results for approximate spherical $k$-designs \cite{brandao2016local, brandao2019models}. We work with the following standard definition of state complexity.
	\begin{definition}
		For a given robustness parameter $\delta \in [0,1]$, an $n$-qubit state $\ket{\psi}$ has weak state complexity $\mathcal{C}_\delta(\psi)  \leq r$ if there exists a circuit $V$ composed of $r$ two-qubit gates of arbitrary connectivity such that $\frac{1}{2}\|\dyad{\psi} - V \dyad{0} V^\dagger\|_1 \leq \delta$.
	\end{definition}
	\noindent We then have the following.
	\begin{corollary}
		Let $\ket{\psi}$ be drawn from an ensemble $\mathcal{E}$ that forms an $\epsilon$-approximate Scrooge $k$-design with additive error. Then, for any $0 < \delta < 1/2k$, the state $\ket{\psi}$ has $\delta$-robust state complexity at least
		\begin{align}
			\mathcal{C}_\delta(\ket{\psi}) > \Omega\left(\frac{k[S_\infty(\rho)-\log k] - \log(1/\eta)}{\log[kS_\infty(\rho)]}\right)
			\label{eq:complexity lower bound app}
		\end{align}
		with probability at least $1 - \frac{\epsilon}{2} - \eta$.
	\end{corollary}
	
	\begin{proof}
		We begin by making use of arguments used in the proof of Proposition 8 in Ref.~\cite{brandao2016local}. First, by Lemma 27 of the same reference, there exists a $\delta'$-covering of cardinality $M \leq {n \choose 2}^r (\frac{10 r}{{\delta'}})^{16r}$ for the set $\mathcal{V}[n, r]$ made up of $n$-qubit states that can be prepared with $r$ two-qubit gates. A $\delta'$-covering (or `$\delta'$-net') of the continuous space $\mathcal{V}[n, r]$ is a finite set $\Upsilon_{\delta'}[n,r]= \{\ket{\phi_i}\}_{i=1}^M$ such that for every $\ket{\chi} \in \mathcal{V}[n, r]$, there exists a $\ket{\phi_i} \in \Upsilon$ such that $\|\dyad{\phi_i} - \dyad{\chi}\|_1 \leq {\delta'}$.
		
		Now let $p_r \coloneqq \text{Pr}_{\psi \sim \mathcal{E}}\big[ \mathcal{C}_\delta(\ket{\psi}) > r \big]$. From the ensemble $\mathcal{E}$, with probability measure $\dif \nu_\mathcal{E}(\psi)$, let us construct a new ensemble $\mathcal{E}'_r$ as follows. To generate samples from $\mathcal{E}'_r$, first sample $\ket{\psi} \sim \mathcal{E}$, and if $\mathcal{C}_\delta(\ket{\psi}) > r$, then we replace $\ket{\psi}$ with $\ket{\phi_0}$, where $\ket{\phi_0}$ is any fixed state contained in $\Upsilon_\delta[n,r]$. This will occur with probability $p_r$. In the other case $\mathcal{C}_\delta(\ket{\psi}) \leq r$, by definition we can find a state $\ket{\chi_\psi}$ that  can be generated by a circuit made up of $r$ two qubit gates, and is $\delta$-close to $\ket{\psi}$. From $\ket{\chi_\psi}$, we can then identify the nearest state in the covering, $\ket{\phi_\psi} \in \Upsilon_{\delta'}[n,r]$, which by the triangle inequality is guaranteed to be $(\delta+\delta')$-close to $\ket{\psi}$. Finally, replace $\ket{\psi}$ with $\ket{\phi_\psi}$. Observe that all states in $\mathcal{E}'_r$ belong to $\Upsilon_\delta[n,r]$, and hence this ensemble has cardinality $\leq M$.
		
		The following inequalities imply that that $\mathcal{E}'_r$ forms a $\epsilon'$-approximate Scrooge $k$-design with additive error, where $\epsilon' \leq \epsilon + 2p_r + k(\delta+\delta')(1-p_r)$.
		\begin{align}
			\|\chi^{(k)}_{\mathcal{E}'_r} - \chi^{(k)}_{\rm Haar}\|_1  &\leq \|\chi^{(k)}_{\mathcal{E}} - \chi^{(k)}_{\rm Haar}\|_1  +\|\chi^{(k)}_{\mathcal{E}'_r} - \chi^{(k)}_\mathcal{E}\|_1  \nonumber\\
			&\leq \epsilon + \left\|\int \dif\nu_\mathcal{E}(\psi) \left(\mathbbm{1}\big[\mathcal{C}_\delta(\ket{\psi}) > r\big](\psi^{\otimes k} - \phi_0^{\otimes k})  + \mathbbm{1}\big[\mathcal{C}_\delta(\ket{\psi}) \leq r\big](\psi^{\otimes k} - \phi_\psi^{\otimes k})\right)\right\| \nonumber\\
			&\leq \epsilon + 2\int \dif\nu_\mathcal{E}(\psi) \mathbbm{1}\big[\mathcal{C}_\delta(\ket{\psi}) > r\big]+ k(\delta+\delta')\int \dif\nu_\mathcal{E}(\psi)  \mathbbm{1}\big[\mathcal{C}_\delta(\ket{\psi}) \leq r\big] \nonumber\\ &= \epsilon + 2p_r + k(\delta+\delta')(1-p_r).
		\end{align}
		In the last inequality, we use the fact that $\|\psi - \phi_\psi\|_1\leq \delta+{\delta'}$ by the arguments given above, and we also invoke the bound $\|\rho^{\otimes k} - \sigma^{\otimes k}\|_1 \leq k\|\rho - \sigma\|$ for any two density operators $\rho$, $\sigma$. (This last inequality follows the identity $(A-B)^{\otimes k} = \sum_{i=0}^{n-1} A^{\otimes i}\otimes (A-B)\otimes B^{\otimes(n-i-1)}$ along with the triangle inequality.)

		We have arrived at an ensemble of cardinality at most $|\Upsilon_{\delta'}[n,r]| \leq {n \choose 2}^r (\frac{10 r}{{\delta'}})^{16r}$ which forms a $\epsilon'$-approximate Scrooge $k$-design with additive error, and so Lemma \ref{lem:DiscreteBound} can be invoked. Recalling that $\epsilon' \leq \epsilon + 2p_r + k(\delta+\delta')(1-p_r)$, we set ${\delta'} = \delta$, and $p_r = 1 -\epsilon/2-\eta$. Using the assumption $k\delta < 1/2$, we then have
		\begin{align}
			{n \choose 2}^r \left(\frac{10 r}{\delta'}\right)^{16r} &\geq \frac{\eta}{2}\frac{2^{kS_\infty(\rho)}}{k!} \nonumber\\ \Rightarrow r &= \Omega\left(\frac{k(S_\infty(\rho)-\log k)-\log \eta}{\log(k S_\infty(\rho))}\right).
		\end{align}
	\end{proof}


\end{document}